\documentclass[11pt]{article}
\usepackage{fullpage}
\usepackage[onehalfspacing]{setspace}
\usepackage[dvipsnames]{xcolor}
\usepackage{amsmath}
\usepackage{amsthm}
\usepackage{amssymb}
\usepackage{newpxtext}
\usepackage{newpxmath}
\usepackage{algorithm}
\usepackage[noend]{algpseudocode}
\usepackage[most]{tcolorbox}
\usepackage{multicol}
\usepackage{caption}
\usepackage{parskip}
\usepackage{tikz}
\usetikzlibrary{arrows.meta,decorations.pathmorphing,shapes.geometric}
\usepackage{booktabs,tabularx,threeparttable,multirow}
\usepackage{subcaption}
\usepackage{xurl}
\PassOptionsToPackage{hyphens}{url}\usepackage{hyperref}
\usepackage[nameinlink,capitalize]{cleveref}
\usepackage{makecell}

\title{Dense Language Generation Made Simple:\\ Deterministic, Randomized, and Multi-Order Algorithms}

\date{}

\newtheorem{theorem}{Theorem}[section]
\newtheorem{lemma}[theorem]{Lemma}

\theoremstyle{definition}
\newtheorem{definition}[theorem]{Definition}
\newtheorem{example}[theorem]{Example}
\newtheorem{remark}[theorem]{Remark}

\newtheorem{fact}[theorem]{Fact}

\AddToHook{env/lemma/begin}{\crefalias{theorem}{lemma}}
\AddToHook{env/conjecture/begin}{\crefalias{theorem}{conjecture}}
\AddToHook{env/corollary/begin}{\crefalias{theorem}{corollary}}
\AddToHook{env/proposition/begin}{\crefalias{theorem}{proposition}}
\AddToHook{env/definition/begin}{\crefalias{theorem}{definition}}
\AddToHook{env/example/begin}{\crefalias{theorem}{example}}
\AddToHook{env/remark/begin}{\crefalias{theorem}{remark}}
\AddToHook{env/question/begin}{\crefalias{theorem}{question}}
\AddToHook{env/condition/begin}{\crefalias{theorem}{condition}}
\AddToHook{env/fact/begin}{\crefalias{theorem}{fact}}
\AddToHook{env/claim/begin}{\crefalias{theorem}{claim}}

\crefname{theorem}{Theorem}{Theorems}
\crefname{lemma}{Lemma}{Lemmas}
\crefname{conjecture}{Conjecture}{Conjectures}
\crefname{corollary}{Corollary}{Corollaries}
\crefname{proposition}{Proposition}{Propositions}
\crefname{definition}{Definition}{Definitions}
\crefname{example}{Example}{Examples}
\crefname{remark}{Remark}{Remarks}
\crefname{question}{Question}{Questions}
\crefname{condition}{Condition}{Conditions}
\crefname{fact}{Fact}{Facts}
\crefname{claim}{Claim}{Claims}

\crefname{Program}{Program}{Programs}
\creflabelformat{Program}{(#2\textup{#1})#3}

\crefname{appendix}{Appendix}{Appendices}
\Crefname{appendix}{Appendix}{Appendices}

\DeclareMathOperator*{\Ex}{\mathbb{E}}
\newcommand{\E}[2]{%
  \ensuremath{%
    \if\relax\detokenize{#1}\relax
      \Ex\left[#2\right]%
    \else
      \Ex\limits_{#1}\left[#2\right]%
    \fi
  }%
}

\newcommand{\Z}{\mathbb{Z}}

\newcommand{\defeq}{\coloneq}
\newcommand{\set}[1]{\ensuremath{\left\{#1\right\}}}
\newcommand{\abs}[1]{\ensuremath{\left\vert#1\right\vert}}

\definecolor{boxc}{rgb}{0.2, 0.7, 0.6}
\definecolor{linkc}{rgb}{0.1, 0.5, 0.4}
\definecolor{citec}{rgb}{0.2, 0.4, 0.7}
\definecolor{urlc}{rgb}{0.2, 0.6, 0.3}
\hypersetup{
    colorlinks=true,
    linkcolor=linkc,
    citecolor=citec,
    urlcolor=urlc
}

\newtcolorbox[auto counter,number within=section]{mybox}[2][]{colback=boxc!8!white,colframe=boxc!60!black,title={#2},#1}

\begin{document}
\begin{titlepage}

\author{%
\begin{tabular}{ccc}
\begin{minipage}[t]{0.3\textwidth}\centering
Ziyi Cai\\
\small Rutgers University\\
\small\href{mailto:zc417@cs.rutgers.edu}{zc417@cs.rutgers.edu}
\end{minipage}
&
\begin{minipage}[t]{0.3\textwidth}\centering
Shuangping Li\\
\small Yale University\\
\small\href{mailto:shuangping.li@yale.edu}{shuangping.li@yale.edu}
\end{minipage}
&
\begin{minipage}[t]{0.3\textwidth}\centering
Yiheng Shen\\
\small Meta\\
\small\href{mailto:yhshen@meta.com}{yhshen@meta.com}
\end{minipage}
\\[8ex]
\multicolumn{3}{c}{%
\begin{tabular}{cc}
\begin{minipage}[t]{0.3\textwidth}\centering
Kangning Wang\\
\small Rutgers University\\
\small\href{mailto:kn.w@rutgers.edu}{kn.w@rutgers.edu}
\end{minipage}
&
\begin{minipage}[t]{0.3\textwidth}\centering
Peng Zhang\\
\small Rutgers University\\
\small\href{mailto:pz149@rutgers.edu}{pz149@rutgers.edu}
\end{minipage}
\end{tabular}}
\end{tabular}%
}

\maketitle
\thispagestyle{empty}
\setcounter{page}{0}

\begin{abstract}
\emph{Language generation in the limit} is a theoretical framework for studying how a generator can learn to produce new valid strings from a stream of positive examples. In this model, an adversary chooses an unknown language from a countable family and enumerates its elements in an arbitrary order, while the generator must eventually output only elements of the language that have not yet appeared in the enumeration. Reliable generation is thus formalized through two eventual guarantees: validity and novelty relative to the observed data. To further quantify the breadth of the generator's outputs, Kleinberg and Wei (FOCS 2025, STOC 2026) introduced lower density as a measure of output coverage. Given an order representing the importance or relevance of possible outputs, lower density is the asymptotic lower bound, as $n$ grows, on the fraction of the first $n$ elements of the target language that the generator outputs before they appear in the data. Kleinberg and Wei showed that $1/2$ is the optimal lower-density guarantee for deterministic algorithms.

We develop a simple and unified framework for obtaining optimal lower-density guarantees. We first give a deterministic algorithm that recovers the optimal guarantee of $1/2$ with a significantly simpler analysis than prior work. We then demonstrate the flexibility of our framework through two extensions. First, against an oblivious adversary, randomization raises the optimal guarantee to $1-1/e$. Second, for any finite collection of orders, the optimal deterministic and randomized guarantees can be achieved simultaneously with respect to every order, so accommodating multiple notions of importance or relevance entails no loss in the optimal guarantee.
\end{abstract}

\newpage

\tableofcontents
\thispagestyle{empty}
\setcounter{page}{0}

\end{titlepage}

\section{Introduction}

Language generation has become one of the central computational tasks in modern AI. Contemporary generative models are expected not merely to recognize or classify text, but also to produce new, meaningful strings. At an abstract level, this task has a simple form: after observing examples from an unknown language, a generator should produce new elements of that language.

Jon Kleinberg and Sendhil Mullainathan \cite{DBLP:conf/nips/KleinbergM24} recently proposed a theoretical model for this task, called \emph{language generation in the limit}, inspired both by modern generative applications and by classical work on language identification \cite{gold1967language,angluin1980inductive}. In this model, an adversary chooses an unknown language $K$ from a countable collection \(\mathcal L\) of languages and enumerates its elements in an arbitrary order. At each step, after observing the next data point revealed by the adversary, the algorithm aims to output a string that belongs to \(K\) and has not appeared in the data observed so far. An algorithm generates in the limit if, after some finite time, all of its outputs satisfy this requirement. Kleinberg and Mullainathan proved a striking positive result: every countable collection of languages admits an algorithm that generates in the limit. Subsequent work has developed learning-theoretic characterizations and stronger notions of generation \cite{DBLP:conf/colt/RamanLT25,DBLP:conf/colt/CharikarP25}, studied the validity--breadth trade-off and notions of representative generation \cite{DBLP:journals/corr/abs-2412-18530,DBLP:conf/stoc/KalavasisMV25,DBLP:conf/icml/PealeRR25,DBLP:journals/corr/abs-2607-23361}, and examined robustness and resource constraints, including noise, contamination, feedback, privacy, and bounded memory \cite{DBLP:conf/icml/RamanR25,DBLP:conf/soda/BaiPZ26,DBLP:journals/corr/abs-2511-07417,DBLP:journals/corr/abs-2604-08504,DBLP:conf/colt/FlammarionPPSS26,DBLP:journals/corr/abs-2606-16077,DBLP:journals/corr/abs-2605-10809,DBLP:journals/corr/abs-2608-05029,DBLP:conf/icml/HannekeKMV26,DBLP:journals/corr/abs-2605-30324}. \cite{DBLP:journals/corr/abs-2605-11302} additionally impose rank-dependent deadlines, studying the trade-off between timely dense generation and hallucination. We refer the reader to the actively maintained website \cite{languagegeneration} for a comprehensive and up-to-date compilation of work on language generation in the limit.

The possibility of successful generation in the limit is conceptually important, but it has a fundamental limitation: the definition guarantees eventual validity but not broad coverage of the target language. For example, if the true language is \(\mathbb{Z}^+\), then a generator that outputs only previously unseen powers of \(10\) can be valid forever, but it produces an extremely sparse subset of the language. Similarly, if the true language is English, a generator that only produces sentences of the form ``this generator is good,'' ``this generator is very good,'' ``this generator is very very good,'' and so on, can avoid invalid outputs while representing only a narrow sliver of English. In this sense, generation in the limit rules out eventual hallucination, but it does not by itself rule out mode collapse.

To quantify this issue, Jon Kleinberg and Fan Wei \cite{DBLP:conf/focs/KleinbergW25} introduced density measures for language generation. Fix an ordering of the ground set of strings, interpreted as an order of importance, priority, or relevance. For a language \(K\), let \(K[n]\) denote the first \(n\) elements of \(K\) under this order. If \(D\subseteq K\) is the set of valid strings that are output by a generator before appearing in the observed data, its lower density in \(K\) is
\[
\liminf_{n\to\infty}
\frac{\bigl\vert D\cap K[n]\bigr\vert}{n}.
\]
Thus, lower density asks whether the generator covers a nonvanishing fraction of the important initial portions of the target language. This mirrors the standard role of density in additive combinatorics and number theory, where it measures how large a subset is within an ordered universe; prominent examples include Szemer\'edi's theorem and the Green--Tao theorem \cite{szemeredi1975sets,green2008primes}. In the language-generation setting, the order specifies which strings should count as early or important.

Kleinberg and Wei showed that positive density is achievable. Their first density result gave an algorithm that generates in the limit and guarantees lower density at least \(1/8\) \cite{DBLP:conf/focs/KleinbergW25}. They also observed a simple upper bound of \(1/2\) for deterministic algorithms: even if the algorithm knows the true language \(K\), an adversary can enumerate half of every prefix before the algorithm has a chance to output those elements. In later work, they matched this upper bound and proved that \(1/2\) is the optimal deterministic lower-density guarantee \cite{DBLP:conf/stoc/KleinbergW26}. These results establish a sharp quantitative answer to the validity--breadth tradeoff in the deterministic single-order setting. Their proofs, however, are technically involved.

\subsection{Our Results}

\begin{table}[t]
  \centering
  \caption{Overview of lower-density guarantees.}
  \label{tab:main-results}

  \begingroup
  \setlength{\tabcolsep}{5pt}
  \renewcommand{\arraystretch}{1.6}

  \begin{tabular}{|c|c|c|c|}
    \cline{3-4}
      \multicolumn{2}{c|}{}
        & \textbf{Single Order}
        & \textbf{Multiple Orders} \\
    \hline

    \multirow{4}{*}[1ex]{\textbf{Deterministic}}
      & \multirow{2}{*}[0.5ex]{\textsc{Lower Bound}}
      & $1/2$
      & $1/2$ \\[-1.5ex]
      & &
        {\footnotesize
          (\cite{DBLP:conf/stoc/KleinbergW26},
          \cref{thm:ps_algo_lower_density,thm:ps_algo_partial_lower_density})}
      &
        {\footnotesize
          (\cref{thm:paired_ps_algo_lower_density})} \\
    \cline{2-4}

      & \multirow{2}{*}[0.5ex]{\textsc{Upper Bound}}
      & $1/2$
      & $1/2$ \\[-1.5ex]
      & &
        {\footnotesize (Trivial)}
      &
        {\footnotesize (Trivial)} \\
    \hline

    \multirow{4}{*}[1ex]{\textbf{Randomized}}
      & \multirow{2}{*}[0.5ex]{\textsc{Lower Bound}}
      & $1-1/e$
      & $1-1/e$ \\[-1.5ex]
      & &
        {\footnotesize
          (\cref{thm:rand_ps_algo_lower_density})}
      &
        {\footnotesize
          (\cref{thm:rand_grouped_ps_algo_lower_density})} \\
    \cline{2-4}

      & \multirow{2}{*}[0.5ex]{\textsc{Upper Bound}}
      & $1-1/e$
      & $1-1/e$ \\[-1.5ex]
      & &
        {\footnotesize
          (\cref{thm:rand_ps_algo_lower_density_upper})}
      &
        {\footnotesize
          (\cref{thm:rand_ps_algo_lower_density_upper})} \\
    \hline
  \end{tabular}

  \endgroup

  \vspace{0.8ex}
  \begin{minipage}{0.92\linewidth}
    \footnotesize
    The upper bounds for the multiple-order settings follow from the corresponding single-order cases.\\
    \cite{DBLP:conf/stoc/KleinbergW26} and
    \cref{thm:ps_algo_partial_lower_density} also establish the $\alpha/2$ lower bound
    in the more general model of \emph{partial enumeration} with parameter $\alpha$.
  \end{minipage}
\end{table}

In this paper, we give a simple route to dense language generation. Our first main result is a simple deterministic algorithm with a short proof achieving the optimal \(1/2\) lower-density guarantee. We then use the same algorithmic framework for two extensions: a randomized algorithm achieving lower density \(1-1/e\), and algorithms whose density guarantees hold simultaneously with respect to multiple orders. Our results are summarized in \cref{tab:main-results}.

\paragraph{Dense language generation made simple (\cref{sec:ps_algo}).} We propose \nameref{box:ps_algo}, which provides the same $1/2$ lower-density guarantee (\cref{thm:ps_algo_lower_density}) as the algorithm of \cite{DBLP:conf/stoc/KleinbergW26}, but with a significantly simpler structure. The proof for our deterministic algorithm achieving $1/2$ lower density has been formalized in Lean by \cite{DBLP:journals/corr/abs-2609-36663}.

\paragraph{The power of randomization (\cref{sec:rand_ps_algo} and \cref{apx:rand_ps_algo_lower_density}).} If the adversary who chooses the order of enumerating the data points can observe and adapt to the algorithm's past outputs, then clearly no algorithm can achieve a lower density better than $1/2$. When the adversary is oblivious (non-adaptive), however, we show that randomization is helpful: a randomized algorithm can achieve lower density $1 - 1/e$ (\cref{thm:rand_ps_algo_lower_density}); this ratio is tight for all algorithms (\cref{thm:rand_ps_algo_lower_density_upper}). To streamline the presentation, the main body (\cref{sec:rand_ps_algo}) presents only a family of algorithms, \nameref{box:rand_ps_algo}, that achieve lower density arbitrarily close to $1 - 1/e$ (\cref{thm:rand_ps_algo_lower_density_weaker}). The complete version with lower-density exactly equal to $1 - 1/e$, \nameref{box:rand_ps_algo_var}, is deferred to \cref{apx:rand_ps_algo_lower_density}.

\paragraph{Having multiple orders does not hurt (\cref{sec:paired_ps_algo} and \cref{apx:paired_ps_algo_lower_density,apx:rand_grouped_ps_algo_lower_density}).} The order in the definition of lower density can represent the relative importance, priority, or relevance of different possible outputs. In many applications, however, there is no single universally accepted order: different users, evaluators, or downstream tasks may rank the same outputs differently. This motivates the requirement for an algorithm to perform well simultaneously with respect to multiple orders. We show that this additional requirement entails no loss in the optimal guarantee. When lower density is computed with respect to one of finitely many orders, we show that a deterministic algorithm achieves lower density $1/2$ against an adaptive adversary (\cref{thm:paired_ps_algo_lower_density}), and that a randomized algorithm achieves lower density $1 - 1/e$ against an oblivious adversary (\cref{thm:rand_grouped_ps_algo_lower_density}), simultaneously for each of these orders. Thus, the optimal ratios remain the same as in the corresponding single-order settings. For the deterministic result, the main body (\cref{sec:paired_ps_algo}) presents only a family of algorithms, \nameref{box:paired_ps_algo}, that achieve lower density arbitrarily close to $1/2$ (\cref{thm:paired_ps_algo_lower_density_weaker}); the full version, \nameref{box:paired_ps_algo_var}, is deferred to \cref{apx:paired_ps_algo_lower_density}. The randomized result is obtained by \nameref{box:rand_grouped_ps_algo} and proved in \cref{apx:rand_grouped_ps_algo_lower_density}.

In independent and concurrent work, Kleinberg and Wei \cite{DBLP:journals/corr/abs-2604-02385} study more general density notions motivated by spatial embeddings of strings. To put it loosely, their notion requires the generated set to be dense not only in initial prefixes but also throughout the embedding. To the best of our knowledge, neither their results nor ours subsumes the other: their work strengthens and generalizes the density requirement, whereas ours establishes simultaneous density guarantees for multiple orders.

\subsection{Technical Overview}

The technical overview below is organized around two ideas. We first explain the core argument behind our simpler deterministic algorithm, and then describe how this argument extends to the randomized and multiple-order settings.

\paragraph{Bounding switch losses suffices.} Our starting point is the original \nameref{box:KM_algo} (\cref{subsec:recap_KM_algo}). Their algorithm maintains a descending chain (under inclusion) of languages that are consistent with all previous announcements, called the \emph{critical chain}. It truncates this chain at a position that grows with the time step and treats the infimum of the resulting chain as its guess for the true language. We identify the obstruction to a $1/2$ density guarantee: some announcements by the adversary are not contained in the current true language guess and therefore falsify it. We call these events, and also the announced integers themselves, \emph{switch losses}. Each previously unannounced switch loss reduces the algorithm's density. Roughly speaking, under the \nameref{box:KM_algo}, the number of such switch losses in a prefix can be asymptotically proportional to the prefix length, making the lower density asymptotically zero.

Our algorithm, the \nameref{box:ps_algo}, improves the \nameref{box:KM_algo} through a more careful truncation policy for the critical chain. The truncation point must still be allowed to grow, but our algorithm lets it grow slowly enough that, whenever a previously unannounced switch loss occurs, the algorithm must have kept the same true language guess for an exponentially long time and, during that period, announced many new integers smaller than the switch loss. As a result, in any prefix, the number of previously unannounced switch losses is logarithmic in the prefix length and therefore negligible.

\paragraph{Block-based algorithms as a flexible meta-algorithm.} A natural starting point for designing dense language-generation algorithms is a toy regime: a single finite language known to both parties. In this toy regime, one asks for algorithms whose density approaches the desired guarantee as the language size tends to infinity. This raises the following question: can such an asymptotically dense toy-regime algorithm be lifted to the original regime with countably many candidate languages, each countably infinite, and an unknown true language?

We answer this question affirmatively through a block-based meta-algorithm. The algorithm adopts the same slow-growing truncation policy for the critical chain as the \nameref{box:ps_algo}, but its key additional feature is that it reserves disjoint blocks of integers on the fly. When it receives an announcement from the adversary, it first checks whether the announced integer lies in a previously reserved block. If so, it follows the corresponding dense strategy for the toy regime. Otherwise, it reserves a new block consisting of the smallest available integers in the current true language guess. Locally, the algorithm is dense inside each block. Globally, however, a prefix may cut through several blocks, leaving them incomplete; for these incomplete blocks, the within-block density guarantee need not apply. The key step in the analysis is therefore to bound the number of incomplete blocks. We do this by relating incomplete blocks to switch losses, whose number can be bounded using an argument analogous to the analysis of the \nameref{box:ps_algo}. See \cref{fig:blocks} for an illustration.

\begin{figure}[t!]\centering
\begin{tikzpicture}[x=0.5cm,y=0.6cm,font=\small]
  \fill[gray!10] (0.6,0) rectangle (13,7.4);
  \draw[densely dashed,thick] (13,-0.3) -- (13,7.4);
  \node[above] at (6.8,7.4) {prefix $[n]$};

  \draw[->,thick] (0.6,0) -- (17.6,0);
  \foreach \x in {1,...,17} {
    \draw (\x,0.12) -- (\x,-0.12);
  }
  \node[below] at (1.5,-0.12) {$1$};
  \node[below] at (12.5,-0.12) {$n$};

  \foreach \r/\els in {1/{1,2,4}, 2/{5,6,9}, 3/{8,11,12}} {
    \node[left,teal!80!black] at (0.6,1.2*\r-0.3) {$B_{\r}$};
    \foreach \x in \els {
      \filldraw[fill=teal!45,draw=teal!80!black] (\x,1.2*\r-0.7) rectangle +(1,0.8);
    }
  }
  \foreach \r/\els in {4/{3,10,14}, 5/{7,15,16}} {
    \node[left,orange!80!black] at (0.6,1.2*\r-0.3) {$B_{\r}$};
    \foreach \x in \els {
      \filldraw[fill=orange!45,draw=orange!80!black] (\x,1.2*\r-0.7) rectangle +(1,0.8);
    }
  }

  \node[align=left,text width=5cm,font=\footnotesize,anchor=west] at (18.2,2.1)
    {\textcolor{teal!80!black}{\textbf{Complete blocks:}} the dense toy-regime algorithm gives density $\gamma$ within each.};
  \node[align=left,text width=4.6cm,font=\footnotesize,anchor=west] at (18.2,5.1)
    {\textcolor{orange!80!black}{\textbf{Incomplete blocks:}} at most $o(n)$ of them.};
\end{tikzpicture}
\caption{Complete and incomplete blocks within a prefix $[n]$. Each row is a single block, whose elements need not be consecutive integers. Hence the density in $[n]$ is approximately $\gamma \cdot \frac{n - o(n)}{n}$, which tends to $\gamma$ as $n \to \infty$.}
\label{fig:blocks}
\end{figure}

We do not state the meta-algorithm explicitly as a standalone result. Instead, we demonstrate its flexibility by instantiating it directly in the randomized (\cref{sec:rand_ps_algo}) and multiple-order (\cref{sec:paired_ps_algo}) settings. In the toy regime above, both settings connect to well-studied problems. The Ranking algorithm \cite{DBLP:conf/stoc/KarpVV90} for online bipartite matching yields a randomized algorithm that is asymptotically optimal, as the language size tends to infinity, against an oblivious adversary; low-crossing partitions \cite{DBLP:conf/compgeom/Welzl89,DBLP:journals/dcg/ChazelleW89,DBLP:journals/dcg/Matousek92} yield an algorithm that is asymptotically optimal, again as the language size tends to infinity, simultaneously under multiple orders. These observations guide the corresponding block-based algorithms in the full regime.

\paragraph{Comparison with Kleinberg--Wei.}
The work most closely related to ours is that of Kleinberg and Wei. Their first density paper gave an algorithm with lower density at least $1/8$ \cite[Theorem~6.12]{DBLP:conf/focs/KleinbergW25}. In follow-up work, they proved the tight bound of $1/2$, and more generally a bound of $\alpha/2$ when the adversary enumerates only a subset of $K$ of lower density $\alpha$ \cite[Theorems~1.3 and~1.6]{DBLP:conf/stoc/KleinbergW26}. We obtain both tight guarantees with a simpler algorithm and a much shorter proof (\cref{thm:ps_algo_lower_density,thm:ps_algo_partial_lower_density}).

Kleinberg and Wei prove the tight bound directly in the partial-enumeration model, where the algorithm cannot in general guess a single language contained in $K$. Their algorithm instead guesses a finite intersection of candidate languages: it intersects the first $k$ consistent languages for each $k$, which gives a descending chain, and selects one intersection from this chain. They show that this guess is eventually contained in $K$ and contains the entire enumerated set at infinitely many time steps \cite[Theorem~2.4 and Lemma~2.5]{DBLP:conf/stoc/KleinbergW26}. Since these time steps can be arbitrarily far apart, density requires an additional layer on top of this guess. Their warm-up argument maintains a priority list of strings, filled through aggressive guesses and tokens, and the tight argument replaces it with pods, batches of reserved strings from which the algorithm outputs one string at a time. The key pod lemma maps each bad adversary string injectively to an earlier pod whose elements all precede it \cite[Lemma~3.6]{DBLP:conf/stoc/KleinbergW26}.

We take a more direct route. As described above, we keep the KM algorithm and let its scope grow patiently, enlarging it only after the focus has stayed the same for $2^\tau$ steps. We never reserve strings and always output the smallest unannounced integer in the focus. We bound the switch losses with a single charging lemma (\cref{lem:ps_algo_switch_loss}). For partial enumeration, we run the same algorithm on the closure of $\mathcal{L}$ under finite intersections and obtain the tight $\alpha/2$ bound (\cref{thm:ps_algo_partial_lower_density}).

Despite these differences, our proof shares some high-level structure with theirs. Both algorithms fall back along the previous chain when the current guess is falsified, and both analyses charge the resulting loss to strings selected when the algorithm last held the guess it falls back to. In this sense, we work directly with the KM algorithm where they use its variant $\mathcal{A}_{\mathrm{acc}}$, the outputs we make while waiting play a role similar to their priority list and pods, our waiting time $2^\tau$ plays a role similar to their token and pod-size parameters, and our charging lemma is closely related to their pod lemma. Several features of their argument, such as the semi-index formulation and the identified intersection, are needed only for partial enumeration and disappear when it is specialized to full enumeration. We therefore view our proof as a short and streamlined treatment of the full-enumeration case, built on related ideas. This simplicity also takes us beyond the deterministic setting: our block-based algorithms reuse the same scope-update rule to achieve lower density $1-1/e$ with randomization and simultaneous guarantees for multiple orders (\cref{sec:rand_ps_algo,sec:paired_ps_algo}).

\section{Preliminaries and Notation}
\paragraph{Notation.}
For any non-negative integer $n$, let $[n]$ denote the set $\set{1, \dots, n}$. For any set $S\subseteq\Z^+$ and any non-negative integer $n$, let $S[n]$ denote the set of the smallest $n$ elements of $S$.

These notations extend naturally to multi-order settings. Throughout this paper, we view a permutation $\sigma$ as a sequence: for every positive integer $j$, $\sigma(j)$ denotes its $j$-th element. Let $\sigma_1, \ldots, \sigma_k$ be $k$ permutations of the universe, which may be countably infinite. For each $i \in [k]$ and non-negative integer $n$, define
\[
[n]_i \defeq \sigma_i([n]) = \set{\sigma_i(1), \ldots, \sigma_i(n)},
\]
the set of the first $n$ integers under the order $\sigma_i$. For any set $S\subseteq\Z^+$, non-negative integer $n$, and index $i \in [k]$, let $S[n]_i$ denote the set of the smallest $n$ elements of $S$ under the order $\sigma_i$.

\paragraph{Language generation in the limit.}
Throughout this paper, a \emph{language} is an infinite subset of $\Z^+$ (the positive integers). The KM model interprets language generation as a game between an \emph{adversary} and an \emph{algorithm} over $\mathcal{L} = \set{L_1, L_2, L_3, \dots}$, a countably infinite collection of languages. The collection $\mathcal{L}$ is fully accessible to both parties.

The adversary chooses a true language $K = L_{i^\star}\in \mathcal{L}$ and keeps it secret. The game then starts and is played in rounds. At each round $t \in \Z^+$, the adversary announces an integer in $K$, and the algorithm then announces a positive integer. The adversary must enumerate $K$ without repetition: every integer in $K$ must be announced by the adversary exactly once, at some finite time step.

When there is no further clarification, we assume that the algorithm is deterministic (based on both $\mathcal{L}$ and the history of the game), and the adversary can be adaptive.

\begin{definition}[Validity]
We say that an algorithm can \emph{generate in the limit} the true language $K$ if there exists an integer $t_0$ such that, in every round $t \geq t_0$, the integer announced by the algorithm belongs to $K$ and has not previously been announced by either party. The integer $t_0$ may depend on the adversary's strategy.
\end{definition}

In the same work \cite{DBLP:conf/nips/KleinbergM24} that introduced the KM model, Kleinberg and Mullainathan propose an algorithm that can always achieve generation in the limit. We will discuss their algorithm further in \cref{subsec:recap_KM_algo}.

\paragraph{Density of language generation.} Let $A$ be the set of positive integers first announced by the adversary ($A$ for ``attacker''), and let $D$ be the set of positive integers first announced by the algorithm ($D$ for ``defender''). For any set $S \subseteq \Z^+$ and any integer $n \in \Z^+$, define $\mu_n(S) = \bigl\vert[n] \cap S\bigr\vert$. Then $\mu_n(A) + \mu_n(D \cap K) = \mu_n(K)$.

\begin{definition}[Lower Density] The lower density of an algorithm is defined by
    \[
    \liminf_{n \to \infty} \frac{\mu_n(D \cap K)}{\mu_n(K)}.
    \]
\end{definition}

We aim to design an algorithm that, regardless of the adversary's strategy, (1) generates in the limit the true language $K$, and (2) achieves a large lower density. A $1/2$ upper bound for lower density is immediate, as the adversary can announce the minimum integer in $K$ that has not been announced by either party (with sufficiently sparse catch-up rounds to enumerate the missing integers). Therefore, one may ask the following question:
\begin{quote}
    Is there an algorithm achieving the optimal lower density $1/2$?
\end{quote}

This question was posed by \cite{DBLP:conf/focs/KleinbergW25}, in which they show a $1/8$ lower bound, and was first answered in the affirmative by \cite{DBLP:conf/stoc/KleinbergW26}. In this paper, we present a significantly simpler proof of this fact.

\subsection{Background for Auxiliary Tools}
\paragraph{Online bipartite matching.}
The input to an online bipartite matching problem is a bipartite graph $G = (U \sqcup V, E)$, where the vertices in $U$ arrive one by one. For every vertex $u \in U$, let $N(u)$ denote the set of neighbors of $u$ in $V$. When a vertex $u \in U$ arrives, its incident edges are revealed, and the algorithm must either match $u$ to one of its currently unmatched neighbors or leave $u$ unmatched. Each decision is irrevocable, and the goal is to maximize the size of the resulting matching. More specifically, we seek to maximize the competitive ratio: the worst-case ratio between the size of the matching produced by the algorithm and the size of the offline maximum matching in $G$.

\begin{algorithm}[H]
\caption{Ranking algorithm for online bipartite matching}
\begin{algorithmic}[1]\label{alg:ranking_algo}
\State Choose a uniformly random permutation $\pi$ of the offline vertices $V$.
\State Initialize $M \gets \emptyset$.
\For{each arriving vertex $u \in U$}
    \State Let $N^\prime$ be the set of unmatched neighbors of $u$ in $V$.
    \If{$N^\prime \neq \emptyset$}
        \State Let $v$ be the vertex in $N^\prime$ with minimum rank under $\pi$.
        \State $M \gets M \cup \{(u,v)\}$.
    \EndIf
\EndFor
\State \Return $M$.
\end{algorithmic}
\end{algorithm}

The online bipartite matching problem was introduced by Karp, Vazirani, and Vazirani \cite{DBLP:conf/stoc/KarpVV90}. In the same work, they proposed the Ranking algorithm (\cref{alg:ranking_algo}) and proved that it achieves the optimal competitive ratio. The algorithm is further analyzed in the work of \cite{MISC:KrohnV07} and \cite{DBLP:conf/soda/GoelM08}.

\begin{theorem}[\cite{DBLP:conf/stoc/KarpVV90,MISC:KrohnV07,DBLP:conf/soda/GoelM08,DBLP:journals/sigact/BirnbaumM08,DBLP:conf/soda/DevanurJK13}]\label{thm:ranking_algo}
Let $m = \abs{U}$ and
\[
f(m) = 1-\left(1-\frac1{m+1}\right)^m.
\]

The Ranking algorithm achieves competitive ratio $f(m)$, which converges to $1 - 1/e$ as $m \to \infty$.
\end{theorem}

This limiting ratio is optimal for the online bipartite matching problem.

\begin{theorem}[\cite{DBLP:conf/stoc/KarpVV90}]\label{thm:online_matching_upper} Let $m = \abs{U}$. No randomized algorithm can achieve an asymptotic competitive ratio better than
\[
1 - \frac1e + o_m(1).
\]
Moreover, for every randomized algorithm on an $m$-by-$m$ bipartite graph, there is a hard instance whose incidence matrix is strictly upper triangular up to a permutation of the columns; equivalently, $N(u_m) = \emptyset$, $N(u_{i + 1}) \subsetneq N(u_i)$, and $\abs{N(u_i) \setminus N(u_{i + 1})} = 1$.
\end{theorem}

\paragraph{Low-crossing partitions.}
Let $m$ be a positive integer, and let $\sigma_1, \dots, \sigma_k$ be $k$ permutations of $[m]$.

We say that a set $G$ is crossed by the cut $(S, [m] \setminus S)$ if
\[
G \cap S \neq \emptyset \qquad \text{and} \qquad G \setminus S \neq \emptyset.
\]
Given a partition of $[m]$ into groups $G_1, G_2, \dots, G_{m/s}$ of equal size $s$, let $c_i(n)$ denote the number of groups crossed by the cut $([n]_i, [m]\setminus [n]_i)$. We call $c_i(n)$ the \emph{crossing number}.

\begin{theorem}[Low-Crossing Partition]\label{thm:low_crossing}
    For every pair of positive integers $s,m$ with $s \mid m$, and every collection of permutations $\sigma_1, \dots, \sigma_k$ of $[m]$, there is a balanced partition $G_1, G_2, \dots, G_{m/s}$ of $[m]$ such that
    \[
    c_i(n) \leq 4 \cdot \left(\frac{m}{s}\right)^{1 - 1/k}
    \]
    for every $i \in [k]$ and every $n \in [m]$.
\end{theorem}

The theorem follows from the standard balanced $k$-d tree construction introduced by \cite{DBLP:journals/cacm/Bentley75}, after representing each element by its ranks in the $k$ permutations. We include a proof in \cref{apx:low_crossing} for completeness. More generally, low-crossing partitions are a central ingredient in the theory of partition trees for geometric range searching; see, for example, \cite{DBLP:conf/compgeom/Welzl89,DBLP:journals/dcg/ChazelleW89,DBLP:journals/dcg/Matousek92}.

\section{The Patient-Scope Algorithm}\label{sec:ps_algo}
\subsection{Recap of the KM Algorithm}\label{subsec:recap_KM_algo}
Our algorithm builds on the work of Kleinberg and Mullainathan \cite{DBLP:conf/nips/KleinbergM24}. We first introduce the terminology needed to present their algorithm.

As the adversary reveals positive samples, we can discard any language that omits one of them.

\begin{definition}[Consistency] We call a language $L_i$ \emph{consistent} if $L_i$ is a superset of all numbers that have been announced by the adversary so far.
\end{definition}

Intuitively, though not rigorously, ``in the limit'' every consistent language must be either the true language $K$ or a proper superset of $K$. Moreover, the language $K$ must be the minimum consistent language under inclusion. Thus, the main idea of the KM algorithm is to maintain a descending chain of consistent languages. We formally define this chain as follows.\footnote{Our definition differs from that of \cite{DBLP:conf/nips/KleinbergM24}, which requires a critical language $L_i$ to be a subset of every consistent language $L_j$ with $j<i$, rather than only those that are critical. See \cref{rmk:criticality} for details.}

\begin{definition}[Criticality]
  The lowest-indexed critical language is the lowest indexed consistent language, if such a language exists. Thereafter, critical languages are defined inductively in increasing index order. Having determined which languages $L_j$ with $j < i$ are critical, we declare $L_i$ to be \emph{critical} if the following two conditions hold:
    \begin{itemize}
        \item $L_i$ is consistent;
        \item $L_i \subseteq L_j$ for every critical language $L_j$ with $j < i$.
    \end{itemize}
\end{definition}

\begin{remark} At every time step, the critical languages form a descending chain under inclusion. Moreover, after the adversary makes an announcement, the languages that cease to be critical form a suffix of the previous critical chain, while every newly critical language has a larger index than every language that remains critical.
\end{remark}

\begin{lemma}\label{lem:criticality} The true language $K$ is critical at all sufficiently large time steps.
\end{lemma}
\begin{proof}[Proof of \cref{lem:criticality}]
There are only finitely many languages with indices lower than that of the true language $K$. Any such language $L$ that is not a superset of $K$ becomes inconsistent once the adversary announces an integer in $K \setminus L$. Therefore, by definition, the true language $K$ is critical after finitely many time steps.
\end{proof}

Since the critical chain can contain infinitely many languages, the algorithm truncates it at a sufficiently long position and uses the running minimum in the truncated chain as its current hypothesis.

\begin{definition}[Scope] The \emph{scope size} is a positive integer $s$ (used by the algorithms). The \emph{scope} is $\set{L_1, L_2, \dots, L_s}$. We denote the scope size at time step $t$ by $s_t$.
\end{definition}

\begin{definition}[Focus] The \emph{focus} is the highest-indexed critical language $F$ within the scope. The \emph{focus index} is the index (in $\mathcal{L}$) of the focus.
\end{definition}

\begin{remark}\label{rmk:criticality} Under the definition of criticality in the work of \cite{DBLP:conf/nips/KleinbergM24}, when a consistent but noncritical language becomes inconsistent, a language with a larger index may become critical within the \emph{existing} scope. This can change the focus even if the current focus remains consistent. Our definition rules out such changes: if the scope does not shrink, the focus can change only if the current focus becomes inconsistent or the scope expands to include a new critical language.
\end{remark}

\begin{mybox}[label={box:KM_algo},nameref={KM algorithm}]{The Kleinberg--Mullainathan (KM) algorithm}
  \begin{itemize}
    \item In each time step $t$, perform the following operations.
    \begin{itemize}
        \item Set the scope size $s_t \gets t$.
        \item Receive a new integer from the adversary, and update the sets of consistent and critical languages as well as the focus accordingly.
        \item Output the smallest integer in the focus $F$ that has not been announced by either party. (If no focus exists, output any number.)
    \end{itemize}
  \end{itemize}
\end{mybox}

\begin{figure}[t!]\centering
\begin{tikzpicture}[font=\small]
    \node[font=\bfseries,align=center] at (5.6,4.15)
      {The Kleinberg--Mullainathan algorithm};

    \draw[->,thick] (1,-0.7) -- (10,-0.7) node[right] {index};

    \foreach \x in {1.5,3.5} {
      \begin{scope}
        \clip ({\x+0.5},0.9) circle[radius=0.4];
        \fill[red!35] ({\x+0.1},0.5) -- ({\x+0.9},1.3) -- ({\x+0.1},1.3) -- cycle;
        \fill[ForestGreen!35] ({\x+0.1},0.5) -- ({\x+0.9},1.3) -- ({\x+0.9},0.5) -- cycle;
      \end{scope}
      \draw[red!80!black] ({\x+0.5},0.9) circle[radius=0.4];
    }

    \foreach \x in {5.5,7.5} {
      \filldraw[fill=ForestGreen!35,draw=ForestGreen!80!black] ({\x+0.5},0.9) circle[radius=0.4];
    }

    \foreach \x in {1.5} {
        \draw[->,thick,red!80] ({\x+0.9},{0.9+0.1}) -- ({\x+2.1},{0.9+0.1});
    }

    \foreach \x in {1.5} {
        \draw[->,thick,ForestGreen!80] ({\x+0.9},{0.9-0.1}) -- ({\x+2.1},{0.9-0.1});
    }

    \foreach \x in {3.5,5.5} {
        \draw[->,thick,ForestGreen!80] ({\x+0.9},0.9) -- ({\x+2.1},0.9);
    }

    \foreach \x in {4.5,6.5} {
        \draw[->,thick,red!80] ({\x+0.9},2.2) -- ({\x+2.1},2.2);
    }

    \draw[->,thick,red!80] ({3.5+0.5+0.24},{0.9+0.32}) -- ({4.5+0.5-0.24},{2.2-0.32});

    \foreach \x in {4.5,6.5,8.5} {
      \filldraw[fill=red!35,draw=red!80!black] ({\x+0.5},2.2) circle[radius=0.4];
    }

    \draw[->,very thick,dotted]
      ({8.5+0.9-0.4},{2.2-0.4}) to[in=0,out=-90]
      node[midway,right=3,font=\small] {focus change} ({7.5+0.9},0.9);

    \node[left,red!80!black,align=center] at (8.5,3.2) {critical languages before time step $t$};

    \node[left,ForestGreen!80!black,align=center] at (8.5,-0.1) {critical languages after time step $t$};

\end{tikzpicture}
\caption{After receiving the $t$-th integer from the adversary, some previously critical languages may become inconsistent. The algorithm then updates the chain of critical languages and changes the focus accordingly.}
\label{fig:KM_algo}
\end{figure}
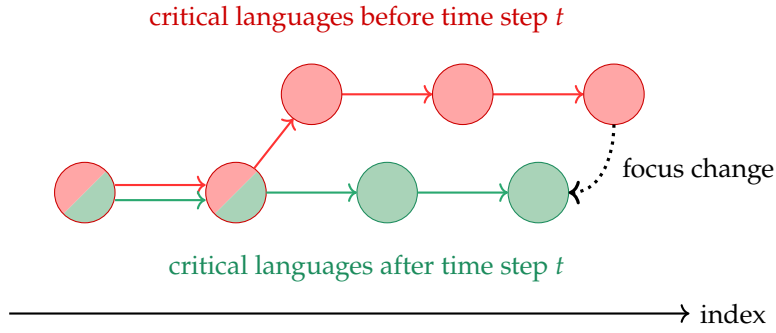

By \cref{lem:criticality}, after finitely many steps the true language $K$ will be in the critical chain. This implies that the focus $F$ must be a subset of $K$. The validity of the KM algorithm follows immediately.

\begin{fact}[\cite{DBLP:conf/nips/KleinbergM24}] The \nameref{box:KM_algo} can generate in the limit the true language $K$.
\end{fact}

Although the \nameref{box:KM_algo} achieves language generation in the limit, it need not achieve any positive lower density. The following example gives an adversarial enumeration under which the algorithm's lower density is $0$.

\begin{example}\label{ex:KM_algo}
Let $L_i = [i] \cup \set{x \in \Z^+ : x \geq 2^i}$. The adversary chooses the true language $K = L_1 = \Z^+$.

\begin{figure}[t!]\centering
\begin{tikzpicture}[font=\small,x=0.6cm,y=0.6cm]
    \tikzset{
        langcell/.style={draw=black!45,line width=0.25pt,minimum width=0.6cm,minimum height=0.6cm,inner sep=0pt},
        missing/.style={langcell,fill=black!12},
        present/.style={langcell,fill=white},
        adv/.style={langcell,fill=red!35},
        alg/.style={langcell,fill=ForestGreen!35}
    }

    \node[font=\bfseries] at (8.5,5.7)
      {The KM algorithm can have zero lower density};

    \foreach \x in {1,...,16} {
        \node at (\x,4.55) {$\x$};
    }

    \foreach \y/\label in {3/$K=L_1$,2/$L_2$,1/$L_3$,0/$L_4$} {
        \node[left] at (0.45,\y) {\label};
        \foreach \x in {1,...,16} {
            \node[missing] at (\x,\y) {};
        }
    }

    \foreach \x in {0,...,16} {
    \node at (\x, -1) {$\vdots$};
    }

    \foreach \x in {1,...,16} {
        \node[present] at (\x,3) {};
    }
    \foreach \x in {1,2,4,5,6,7,8,9,10,11,12,13,14,15,16} {
        \node[present] at (\x,2) {};
    }
    \foreach \x in {1,2,3,8,9,10,11,12,13,14,15,16} {
        \node[present] at (\x,1) {};
    }
    \foreach \x in {1,2,3,4,16} {
        \node[present] at (\x,0) {};
    }

    \foreach \y in {1, 2, 3, 4} {
        \node[adv] at (\y,{4-\y}) {$\y$};
    }

    \foreach \y in {1, 2, 3, 4} {
        \pgfmathtruncatemacro{\py}{2^\y}
        \node[alg] at ({2^\y},{4-\y}) {$\py$};
    }

    \node[adv] at (0,-2.35) {};
    \node[anchor=west] at (0.55,-2.35) {the adversary's announcement};
    \node[missing] at (11,-2.35) {};
    \node[anchor=west] at (11.55,-2.35) {integers not in $L_i$};
    \node[alg] at (0,-3.75) {};
    \node[anchor=west] at (0.55,-3.75) {the algorithm's announcement};
    \node[present] at (11,-3.75) {};
    \node[anchor=west] at (11.55,-3.75) {other integers in $L_i$};
\end{tikzpicture}
\caption{The example $L_i=[i]\cup\set{x\in\Z^+ : x\geq 2^i}$ and $K = L_1 = \Z^+$. The $t$-th row represents the focus $L_t$ at time step $t$; highlighted cells indicate the integers announced by the adversary and by the algorithm. Gray cells are integers not contained in the corresponding language.}
\label{fig:KM_zero_density_example}
\end{figure}
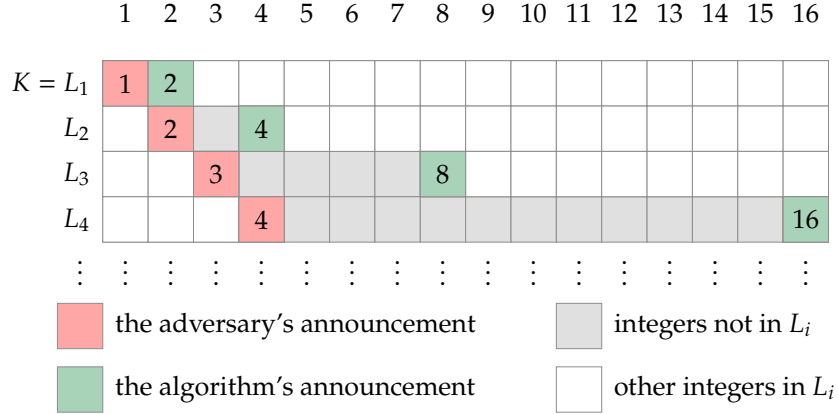

At time step $t$, the adversary announces $t$. After the first $t$ rounds, $L_t$ is consistent; moreover, every $L_j$ with $2 \leq j < t$ is inconsistent because $j+1 \notin L_j$, while $L_1$ contains $L_t$. Thus, $L_t$ is critical and, since the scope size is $t$, is the focus of the KM algorithm. The algorithm therefore outputs $2^t$, the smallest integer in $L_t$ that has not yet been announced by either party. Thus, the integers first announced by the algorithm are precisely the powers of two with positive exponents, and the lower density is
\[
\liminf_{n \to\infty}\frac{\mu_n(D\cap K)}{\mu_n(K)} = \liminf_{n \to\infty}\frac{\mu_n(D)}{\mu_n(\Z^+)} = \liminf_{n \to\infty}\frac{\lfloor\log_2 n\rfloor}{n} = 0.
\]
\end{example}

Before introducing our algorithm, we first examine why the \nameref{box:KM_algo} can have zero lower density, as illustrated by \cref{ex:KM_algo}. We use the following terminology in the analysis.

\begin{definition}[Switch Loss]
A \emph{switch loss} is an integer announced by the adversary that is not in the focus from the end of the previous time step. Equivalently, its announcement makes that focus inconsistent.
\end{definition}

\begin{fact}\label{fac:KM_one_to_one}
Apart from switch losses, any integer first announced by the adversary lies in the focus from the end of the previous time step. Moreover, it is at least as large as the \nameref{box:KM_algo}'s output in the previous time step.
\end{fact}
\begin{proof}[Proof of \cref{fac:KM_one_to_one}] If the adversary's announcement is not a switch loss, then it does not make the focus from the previous time step inconsistent, even though the focus may change because $s_t$ increases. Thus, the announced integer lies in the previous focus. The claim then follows because the \nameref{box:KM_algo} always outputs the smallest integer in the focus $F$ that has not yet been announced by either party.
\end{proof}

Suppose the algorithm has run for sufficiently many steps so that all subsequent outputs lie in the true language $K$. By \cref{fac:KM_one_to_one}, within any prefix $[n]$, and excluding switch losses, the number of integers first announced by the algorithm is at least the number of integers first announced by the adversary. Therefore, if a prefix $[n]$ contains too many switch losses, say a number proportional to $n$, then the algorithm's lower density can be small.

In \cref{ex:KM_algo}, every integer that is not a power of two is a switch loss, and hence the lower density is zero. In fact, even if we replace $s_t = t$ by any monotone and unbounded function of $t$, such as $\lceil\log t\rceil$, one can still construct a similar example with zero lower density.

\subsection{Our Algorithm}
The main idea of our algorithm is to let the scope grow very slowly, and even shrink when necessary, while still ensuring that it eventually contains the true language $K$.

\begin{mybox}[label={box:ps_algo},nameref={patient-scope algorithm}]{The patient-scope algorithm}
  \begin{itemize}
    \item Initially, at time $t = 0$, set the scope size $s_0 = 1$, and set the focus-change count $\tau = 1$.
    \item In each time step $t$, perform the following operations.
    \begin{itemize}
        \item Set the scope size $s_t = s_{t - 1}$.
        \item Receive a new integer from the adversary, and update the sets of consistent and critical languages accordingly.
        \item If the focus becomes inconsistent, run the \nameref{box:bt_algo}.
        \item Otherwise, (if the focus is still consistent):
        \begin{itemize}
            \item If the focus has not changed during the previous $2^\tau$ time steps, increase the scope size $s_t$ by one and update the focus accordingly. If the focus changes, increase the focus-change count $\tau$ by one.
        \end{itemize}
        \item Announce the smallest integer in the focus $F$ that has not been announced by either party.
    \end{itemize}
  \end{itemize}
\end{mybox}

\begin{mybox}[label={box:bt_algo},nameref={backtracking algorithm}]{The backtracking algorithm}
\begin{itemize}
            \item Increase the focus-change count $\tau$ by one.
            \item If there is no consistent language within the scope $\{L_1, L_2, \ldots, L_{s_t}\}$, increase the scope size $s_t$ until there is one.\footnote{Note that the true language $K$ is always consistent, and so (by the proof of \cref{lem:ps_algo_valid}) this ``if'' is only true for a finite number of times.}
            \item Otherwise, find the highest-indexed language $L_i$ within the scope that was critical before the announcement in this time step and is still critical after the announcement. If there is no such language $L_i$, choose $L_i$ to be the lowest-indexed consistent language in the scope instead.\footnote{Again, note that because of the true language $K$, this ``if'' is only true for a finite number of times.} Reduce the scope size $s_t$ to be equal to $i$, the index of $L_i$.
        \end{itemize}
\end{mybox}

\begin{figure}[t!]\centering
\begin{tikzpicture}[font=\small]
    \node[font=\bfseries,align=center] at (5.6,4.15)
      {The patient-scope algorithm};

    \draw[->,thick] (1,-0.7) -- (10,-0.7) node[right] {index};

    \foreach \x in {1.5,3.5} {
      \begin{scope}
        \clip ({\x+0.5},0.9) circle[radius=0.4];
        \fill[red!35] ({\x+0.1},0.5) -- ({\x+0.9},1.3) -- ({\x+0.1},1.3) -- cycle;
        \fill[ForestGreen!35] ({\x+0.1},0.5) -- ({\x+0.9},1.3) -- ({\x+0.9},0.5) -- cycle;
      \end{scope}
      \draw[red!80!black] ({\x+0.5},0.9) circle[radius=0.4];
    }

    \foreach \x in {5.5,7.5} {
      \begin{scope}
        \clip ({\x+0.5},0.9) circle[radius=0.4];
        \fill[ForestGreen!15] ({\x+0.1},0.5) rectangle ({\x+0.9},1.3);
        \foreach \offset in {-0.6,-0.4,...,0.6} {
          \draw[ForestGreen!80!black,dashed] ({\x+0.1+\offset},0.5) -- ({\x+0.9+\offset},1.3);
        }
      \end{scope}
      \draw[ForestGreen!80!black,dashed] ({\x+0.5},0.9) circle[radius=0.4];
    }

    \foreach \x in {1.5} {
        \draw[->,thick,red!80] ({\x+0.9},{0.9+0.1}) -- ({\x+2.1},{0.9+0.1});
    }

    \foreach \x in {1.5} {
        \draw[->,thick,ForestGreen!80] ({\x+0.9},{0.9-0.1}) -- ({\x+2.1},{0.9-0.1});
    }

    \foreach \x in {3.5,5.5} {
        \draw[->,dashed,thick,ForestGreen!80] ({\x+0.9},0.9) -- ({\x+2.1},0.9);
    }

    \foreach \x in {4.5,6.5} {
        \draw[->,thick,red!80] ({\x+0.9},2.2) -- ({\x+2.1},2.2);
    }

    \draw[->,thick,red!80] ({3.5+0.5+0.24},{0.9+0.32}) -- ({4.5+0.5-0.24},{2.2-0.32});

    \foreach \x in {4.5,6.5,8.5} {
      \filldraw[fill=red!35,draw=red!80!black] ({\x+0.5},2.2) circle[radius=0.4];
    }

    \draw[->,very thick,dotted]
      ({8.5+0.9-0.4},{2.2-0.4}) to[in=120,out=-90]
      node[midway,above=3,right=3,font=\small] {focus change} ({3.5+0.9-0.4},{0.9+0.4});

    \node[left,red!80!black,align=center] at (8.5,3.2) {critical languages before time step $t$};

    \node[left,ForestGreen!80!black,align=center] at (8.5,-0.1) {critical languages after time step $t$};

\end{tikzpicture}
\caption{After receiving the $t$-th integer from the adversary, some previously critical languages may become inconsistent. Starting from the old focus, the algorithm backtracks along the old critical chain to find the highest-indexed consistent language. It then sets the scope to that index, thereby truncating the critical chain at that position. The dashed vertices indicate the part of the new critical chain that would remain if the scope were unchanged from the previous time step.}
\label{fig:ps_algo}
\end{figure}
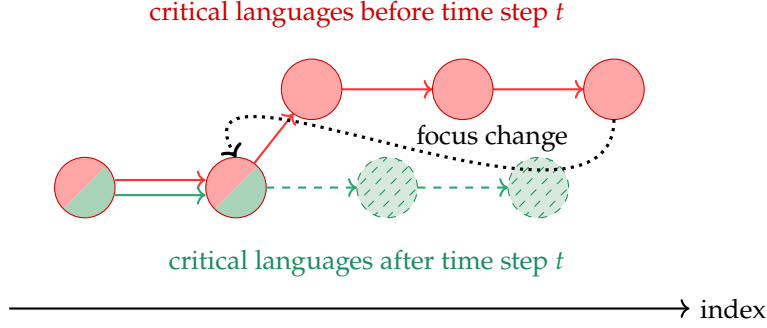

The \nameref{box:ps_algo} differs from the \nameref{box:KM_algo} in two main respects:
\begin{enumerate}
  \item Instead of advancing the scope by one index at the beginning of each time step, the algorithm keeps the scope fixed until it has announced a sufficiently large number of integers. These announcements help reduce the density of switch losses.
  \item When a switch loss occurs, the algorithm invokes the \nameref{box:bt_algo}, starting from the old focus and backtracking along the old critical chain to find the highest-indexed consistent language. It then truncates the critical chain at that position by setting the scope to that index, as illustrated in \cref{fig:ps_algo}. This is analogous to the Knuth--Morris--Pratt (KMP) algorithm \cite{DBLP:journals/siamcomp/KnuthMP77} for pattern matching: when a mismatch occurs, the KMP algorithm backtracks using the failure function until it finds the first compatible partial match.
\end{enumerate}

We first show that, although the scope of the \nameref{box:ps_algo} can shrink, this does not affect the algorithm's validity.

\begin{lemma}[Validity]
\label{lem:ps_algo_valid}
The \nameref{box:ps_algo} can generate in the limit the true language $K$.
\end{lemma}
\begin{proof}
It is clear from its description that the \nameref{box:ps_algo} only announces numbers that have not been previously announced by either party.

Recall (from \cref{lem:criticality}) that the true language $K = L_{i^\star}$ is critical at all sufficiently large time steps. After that point, the scope size cannot go from being at least $i^\star$ to being strictly less than $i^\star$. Moreover, the scope size eventually becomes at least $i^\star$, since the scope size can only decrease when a language within the scope becomes inconsistent.

When $K$ is critical and within the scope, every number announced by the \nameref{box:ps_algo} is in $K$. In particular, there is a time step $T^\star$ such that every focus after time step $T^\star$ is a subset of $K$.
\end{proof}

Now we have established that the true language $K$ is critical at all sufficiently large time steps, and that the \nameref{box:ps_algo} always announces a new integer in the true language $K$ once $K$ is critical and within the scope. To simplify our exposition, we map $K$ to $\Z^+$ (with the ordering in $K$ preserved) and assume $K = \Z^+$. Formally, we apply the partial map that sends each integer $x \in K$ to its rank in $K$ and is undefined outside $K$.

We will use the following observation, analogous to \cref{fac:KM_one_to_one}.

\begin{fact}\label{fac:ps_one_to_one}
Apart from switch losses, any integer first announced by the adversary lies in the focus from the end of the previous time step. Moreover, it is at least as large as the output of the \nameref{box:ps_algo} in the previous time step.
\end{fact}

We omit the proof of \cref{fac:ps_one_to_one}, since it is analogous to the proof of \cref{fac:KM_one_to_one}.

Next, we bound the number of previously unannounced switch losses (i.e.\@ those first announced by the adversary) in every prefix.

\begin{lemma}[Charging Lemma]\label{lem:ps_algo_switch_loss}
For any prefix $[n]$, the number of previously unannounced switch losses in $[n]$ that occur after time step $T^\star$ is at most $\log_2 n$.
\end{lemma}

\begin{proof}[Proof of \cref{lem:ps_algo_switch_loss}]
  
\begin{figure}[t!]\centering
\begin{tikzpicture}[font=\small]
    \node[font=\bfseries,align=center] at (6.1,4.35)
      {Charging switch losses};

    \draw[->,thick] (1,-0.9) -- (11.2,-0.9) node[right] {index};

    \foreach \x/\name in {2.3/$L_i$,5.15/{},8/$L_j$} {
      \filldraw[fill=ForestGreen!35,draw=ForestGreen!80!black] (\x,0.5) circle[radius=0.4];
      \node[below=7] at (\x,0.1) {\name};
    }
    \draw[->,thick,ForestGreen!80] (2.7,0.5) -- (4.75,0.5);
    \draw[->,thick,ForestGreen!80] (5.55,0.5) -- (7.6,0.5);

    \node[align=center,font=\footnotesize,ForestGreen!80!black] at (3.6,-0.25)
      {spent $2^{\tau_1}$ steps\\with focus $L_i$};

    \draw[->,very thick,dotted,red!80!black]
      (7.65,0.9) to[out=160,in=20]
      (2.65,0.9);
    \node[font=\footnotesize,red!80!black] at (5.2,1.08) {switch loss $\ell_1$};

    \foreach \x/\name in {4.1/$L_{i'}$,6.95/{},9.8/$L_{j'}$} {
      \filldraw[fill=ForestGreen!35,draw=ForestGreen!80!black] (\x,{2.75-0.3}) circle[radius=0.4];
      \node[above=7] at (\x,{3.15-0.3}) {\name};
    }
    \draw[->,thick,ForestGreen!80] (4.5,{2.75-0.3}) -- (6.55,{2.75-0.3});
    \draw[->,thick,ForestGreen!80] (7.35,{2.75-0.3}) -- (9.4,{2.75-0.3});

    \node[align=center,font=\footnotesize,ForestGreen!80!black] at (5.5,{3.35-0.3})
      {spent $2^{\tau_2}$ steps\\with focus $L_{i'}$};

    \draw[->,very thick,dotted,red!80!black]
      (9.45,{2.38-0.3}) to[out=200,in=-20]
      (4.45,{2.38-0.3});
    \node[font=\footnotesize,red!80!black] at (6.8,{2.15-0.3}) {switch loss $\ell_2$};

    \draw[->,thick,ForestGreen!80]
      (2.5,0.85) -- (3.84,{2.44-0.3});
\end{tikzpicture}
\caption{The charging argument for switch losses. If a switch loss $\ell$ moves the focus from $L_j$ back to $L_i$, then before the focus could advance from $L_i$ toward $L_j$, the algorithm must have spent $2^{\tau^\prime}$ consecutive steps with focus $L_i$, where $\tau^\prime$ is the value of $\tau$ during that period. The outputs from those steps are charged to $\ell$.}
\label{fig:ps_switch_loss}
\end{figure}
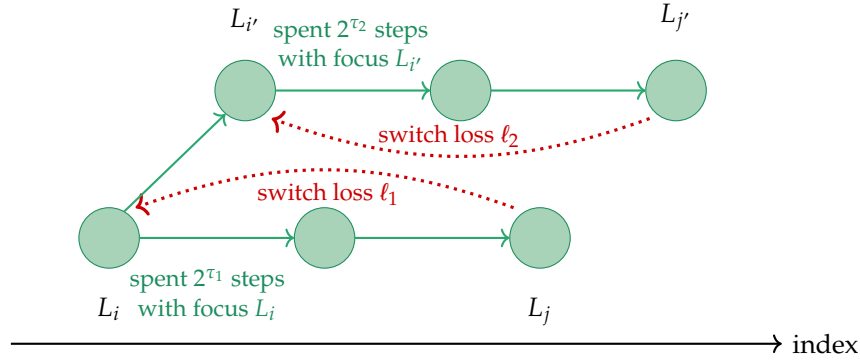

  Consider a previously unannounced switch loss $\ell$ caused by a focus change from $L_j$ to $L_i$, as shown in \cref{fig:ps_switch_loss}. Then $i < j$ and $L_j \subseteq L_i \subseteq K$. During the progression of the focus from $L_i$ to $L_j$, the \nameref{box:ps_algo} must have spent $2^{\tau^\prime}$ consecutive time steps with focus $L_i$, where $\tau^\prime$ is the value of $\tau$ during that period. Since integer $\ell$ is previously unannounced and $\ell \in L_i$, the integers announced by the algorithm during these $2^{\tau^\prime}$ steps are smaller than $\ell$, and hence smaller than $n$. We \emph{charge} the switch loss $\ell$ to these $2^{\tau^\prime}$ integers.
  
	The focus-change counts $\tau$ are distinct across different switch losses, even though different switch losses may correspond to the same $L_i$. Moreover, the sets of integers charged to different switch losses are mutually disjoint.

	Therefore, if there are $w$ switch losses $\ell_1, \ell_2, \dots, \ell_w$ in $[n]$, with corresponding focus-change counts $\tau_1, \tau_2, \dots, \tau_w$, then disjointness of the charged sets implies
\[
n \geq 2^{\tau_1} + 2^{\tau_2} + \cdots + 2^{\tau_w}.
\]
Since $\tau_1, \tau_2, \dots, \tau_w$ are distinct positive integers, we have $n \geq 2^1 + \cdots + 2^w \geq 2^w$. Taking logarithms of both sides completes the proof.
\end{proof}

We can now prove the lower-density guarantee for the algorithm.

\begin{theorem}[Lower Density]\label{thm:ps_algo_lower_density}
  The \nameref{box:ps_algo} achieves the optimal lower density of $1/2$.
\end{theorem}
\begin{proof}[Proof of \cref{thm:ps_algo_lower_density}] Fix an arbitrary prefix $[n]$, and let $M_0$ be the number of integers first announced by the adversary up to time step $T^\star + 1$. By \cref{lem:ps_algo_valid}, $M_0$ is finite and does not depend on $n$. Consider any integer in $[n]$ first announced by the adversary at a time step $t > T^\star + 1$.
  \begin{itemize}
    \item If such an integer is a switch loss, then by \cref{lem:ps_algo_switch_loss}, the number of such integers is at most $\log_2 n$.
    \item Otherwise, by \cref{fac:ps_one_to_one}, it is at least as large as the algorithm's output at time step $t - 1$. This output lies in a focus contained in $K$, and hence in $[n]$. Therefore, each such integer can be mapped to a unique smaller integer in $[n]$ that was first announced by the algorithm.
  \end{itemize}

  These two cases imply that, for every integer $n$,
  \[
  \mu_n(D \cap K) \geq \mu_n(A) - M_0 - \log_2n.
  \]
  It follows that
  \begin{align*}
\liminf_{n \to \infty}\frac{\mu_n(D\cap K)}{\mu_n(K)}
&=\liminf_{n \to \infty}\frac{2\mu_n(D\cap K)}{2n} \\
&\geq \liminf_{n \to \infty}\frac{\mu_n(D \cap K) + \mu_n(A) - M_0 - \log_2n}{2n} \\
&=\liminf_{n \to \infty}\frac{\mu_n(K) - M_0 - \log_2n}{2n} \\
&=\liminf_{n \to \infty}\frac{n - M_0 - \log_2n}{2n}\\
&= \frac12.\qedhere
  \end{align*}
\end{proof}

\begin{remark}
The same argument extends to the setting in which the algorithm may announce $q$ integers at each time step: it simply announces the $q$ smallest unannounced integers in its focus. Apart from a vanishing fraction of switch losses, each integer first announced by the adversary can then be matched with $q$ distinct smaller integers first announced by the algorithm, and these matches are disjoint. Thus, the algorithm achieves lower density $q/(q+1)$.
\end{remark}

\begin{remark}
Kleinberg and Wei \cite{DBLP:conf/focs/KleinbergW25} call an index-based generation algorithm \emph{accurate} at a time step if its current language hypothesis equals the true language $K$, and \emph{accurate infinitely often} if this occurs at infinitely many time steps. Viewing the focus as its current hypothesis, the \nameref{box:ps_algo} also has this property. Once $K$ is permanently critical and within the scope, consider any time the focus moves from $K$ to a proper subset, and let $L$ be the first language after $K$ on the resulting critical chain. Since the adversary enumerates all of $K$, it eventually announces an integer in $K\setminus L$, causing $L$ and every later language on the chain to become inconsistent; the \nameref{box:bt_algo} then returns the focus to $K$. Thus, the focus either eventually remains equal to $K$ or returns to $K$ after every excursion, and in either case the algorithm is accurate infinitely often.
\end{remark}

\paragraph{Application: language generation from partial enumeration.}
Our approach also recovers the optimal lower density in the \emph{partial enumeration setting}, first established by \cite{DBLP:conf/stoc/KleinbergW26}. In this setting, the adversary chooses the true language $K$ from the collection $\mathcal{L}$, together with an infinite subset $K^\prime \subseteq K$ that need not belong to $\mathcal{L}$. The adversary keeps both sets secret and enumerates only the partial language $K^\prime$. The algorithm proceeds as before: at each time step, it receives an announcement from the adversary and responds with an integer that has not yet been announced by either party. We still seek to maximize the lower density of the algorithm, defined by
\[
\liminf_{n \to \infty}\frac{\mu_n(D\cap K)}{\mu_n(K)}.
\]
We show that if $K^\prime$ has lower density $\alpha$ in $K$, then running the \nameref{box:ps_algo} on the collection of all finite intersections of languages in $\mathcal{L}$ achieves lower density $\alpha/2$, matching the tight bound of \cite{DBLP:conf/stoc/KleinbergW26}. We defer the details to \cref{apx:ps_algo_partial}.

\section{Oblivious Adversaries and Randomized Algorithms}\label{sec:rand_ps_algo}
In this section, we describe a randomized variant of \nameref{box:ps_algo}. Randomization is useful only if the adversary cannot react to the algorithm's random choices: against an adaptive adversary, once the algorithm's random choice at a step is realized, the adversary can choose an optimal response. We therefore assume that the adversary is \emph{oblivious}: it fixes in advance a deterministic announcement order for the true language $K$, which remains unknown to the algorithm.

\begin{definition}[Lower Density of a Randomized Algorithm]
    The lower density of a randomized algorithm is
    \[
    \liminf_{n \to \infty}\frac{\E{}{\mu_n(D \cap K)}}{\mu_n(K)}.
    \]
\end{definition}

We show that, against a non-adaptive adversary, randomization improves the best achievable lower density.

\begin{theorem}[Lower Density of a Randomized Algorithm (Lower Bound)]\label{thm:rand_ps_algo_lower_density}
There is a randomized algorithm that, against every non-adaptive adversary, achieves lower density $1 - 1/e \approx 0.63$.
\end{theorem}

To streamline the presentation, we defer the proof of \cref{thm:rand_ps_algo_lower_density} to \cref{apx:rand_ps_algo_lower_density}. In this section, we instead prove a slightly weaker statement that preserves the main idea of the algorithm.

\begin{theorem}[Lower Density of a Randomized Algorithm (Weaker Lower Bound)]\label{thm:rand_ps_algo_lower_density_weaker}
For every positive integer $m$, there is a randomized algorithm that, against every non-adaptive adversary, achieves lower density
\[
\left(1 - \frac1{m + 1}\right) \cdot f(m + 1),
\]
where
\[
f(m) = 1-\left(1-\frac1{m+1}\right)^m.
\]
This lower density tends to $1 - 1/e$ as $m \to \infty$.
\end{theorem}

We also provide a matching $1-1/e$ upper bound in \cref{subsec:rand_ps_algo_lower_density_upper}.

\begin{theorem}[Lower Density of a Randomized Algorithm (Upper Bound)]\label{thm:rand_ps_algo_lower_density_upper}
No randomized algorithm can achieve lower density greater than $1 - 1/e$ against every non-adaptive adversary.
\end{theorem}

\subsection{Warm-Up}
We illustrate the idea with the single-language game $\mathcal{L} = \set{\Z^+}$.

The algorithm fixes an integer $m$, partitions $\Z^+$ into blocks of size $m$, namely $[1, m], [m + 1, 2m], \ldots$, and associates each block with an independent uniformly random permutation of its elements.

After the adversary announces a number, the algorithm identifies the corresponding block and announces the first unannounced element in that block's permutation.

\begin{lemma}\label{lem:block_vs_per}
Within each block, the expected number of integers first announced by the algorithm is at least \[(m - 1) \cdot f(m)\] out of the $m$ integers in the block, regardless of the adversary's announcement order.
\end{lemma}

Then by \cref{lem:block_vs_per} and linearity of expectation, the algorithm achieves lower density $(1 - 1/m) \cdot f(m)$.

\begin{figure}[t!]\centering
\begin{tikzpicture}[
    scale=0.9,
    every node/.style={circle, draw, minimum size=8mm, inner sep=0pt},
    edge/.style={gray!80, line width=0.5pt, thick}
]

\def\n{7}

\foreach \i in {1,...,\n} {
    \node[orange!80!black, fill=orange!10] (u\i) at (0,-\i) {$u_{\i}$};
}

\foreach \j in {1,...,\n} {
    \node[teal!80!black, fill=teal!15] (v\j) at (4,-\j) {$\j$};
}

\foreach \i [evaluate=\i as \iplusone using int(\i+1)] in {1,...,\numexpr\n-1\relax} {
    \foreach \j in {\iplusone,...,\n} {
        \draw[edge] (u\i) -- (v\j);
    }
}

\node[draw=none, rectangle, orange!80!black] at (0,0) {online};
\node[draw=none, rectangle, teal!80!black] at (4,0) {offline};

\draw[->, thick, orange!80!black] (-1,-1) -- (-1,-\n)
    node[midway, left=0.3cm, draw=none, rectangle] {arrival order};

\end{tikzpicture}
\caption{The online bipartite matching instance for $m = 7$, induced by the announcement order in which the adversary announces $t$ at time step $t$. The online vertex $u_t$ represents time step $t$, and the offline vertices represent the integers in $[m]$. Since the integers $1,\ldots,t$ have already been announced by the adversary at time step $t$, the vertex $u_t$ is adjacent exactly to the remaining integers $t + 1,\ldots,m$.}
\label{fig:ranking_algo}
\end{figure}
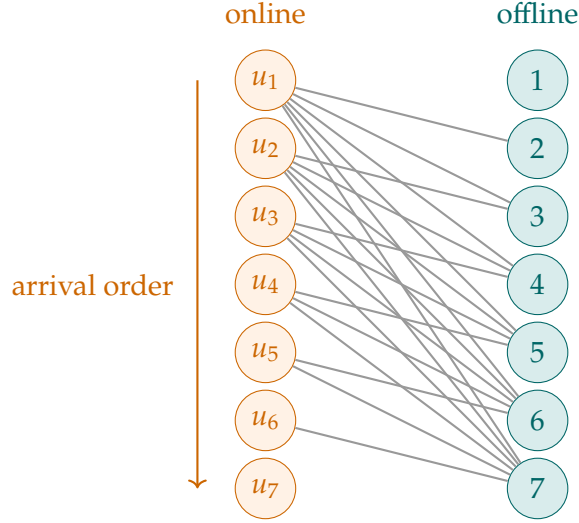
\begin{proof}[Proof of \cref{lem:block_vs_per}]
We connect this problem to online bipartite matching as follows. Consider a bipartite graph $G = (U \sqcup V, E)$, where $V = [m]$ represents the integers to be announced and $U$ represents the $m$ time steps. The vertices $u_t \in U$ arrive online, one at a time. Given the adversary's predetermined announcement order, the vertex $u_t$ is connected to the $m - t$ integers that have not yet been announced by the adversary at time step $t$; see \cref{fig:ranking_algo} for the case $m = 7$. Each matched edge corresponds to an integer first announced by the algorithm at the associated time step.

Under this correspondence, the algorithm's announcements within a block are exactly the matches produced by the Ranking algorithm (\cref{alg:ranking_algo}) with the same random permutation. Note that the offline maximum matching has size $m - 1$, since one can always match $u_t$ to the adversary's announcement at time step $t + 1$ when $t < m$. By \cref{thm:ranking_algo}, the expected size of the online matching produced by the Ranking algorithm, or equivalently the expected number of integers first announced by the algorithm, is at least $(m - 1)\cdot f(m)$.
\end{proof}

In the \nameref{box:ps_algo}, the focus may change infinitely many times. We therefore adapt this toy idea by creating the blocks dynamically.

\subsection{The Randomized Patient-Scope Algorithm}
We fix a positive integer $m$, called the \emph{block size}. The algorithm maintains a family of pairwise disjoint blocks $B_1, B_2, \dots,$ each of size $m$. Each block $B_r$ is equipped with an independent uniformly random permutation $\pi_r$.

\begin{definition}[Reserved]
A number is called \emph{reserved} if it belongs to one of the blocks $B_r$.
\end{definition}

\begin{definition}[Available]
A number is called \emph{available} if it has not yet been announced by either party and has not been reserved by the algorithm. Let $W$ denote the set of currently available numbers.
\end{definition}

\begin{mybox}[label={box:rand_ps_algo},nameref={randomized patient-scope algorithm}]
{The randomized patient-scope algorithm}
\begin{itemize}
    \item Fix a block size $m\in\mathbb{Z}^+$. Let $N_B=0$ be the current number of blocks.
    \item Initially, at time $t=0$, set the scope size $s_0=1$, set the focus-change
    count $\tau=1$, and create no blocks.
    \item In each time step $t$, perform the following operations.
    \begin{itemize}
        \item Set the scope size $s_t = s_{t - 1}$.
        \item Receive a new integer from the adversary, and update the sets of consistent and critical languages accordingly.
        \item If the focus becomes inconsistent, run the \nameref{box:bt_algo}.
        \item Otherwise, if the focus is still consistent and the algorithm has reserved $2^\tau$ blocks since the last focus change, increase the scope size by one and update
        the focus accordingly. If the focus changes, increase $\tau$ by one.
        \item Let $x_t$ be the adversary's announced number in this step.
        \begin{itemize}
            \item If $x_t\in B_r$ for some block $B_r$ and $B_r$ contains a number not yet announced by either party,
            output the first such number in the permutation $\pi_r$.
            \item Otherwise, let \[B_{N_B+1}\defeq (W \cap F)[m].\] Namely, block $B_{N_B+1}$ consists of the $m$ smallest available elements of the current focus $F$. Generate an independent uniformly random permutation $\pi_{N_B+1}$ of $B_{N_B+1}$, mark all elements of
            $B_{N_B+1}$ as reserved, and output the first number of $\pi_{N_B+1}$. Increase $N_B$ by $1$.  
        \end{itemize}
    \end{itemize}
\end{itemize}
\end{mybox}

We first show the validity of the \nameref{box:rand_ps_algo}.

\begin{lemma}[Validity]\label{lem:rand_ps_algo_validity}Regardless of the realization of random permutations, the \nameref{box:rand_ps_algo} can always generate in the limit the true language $K$.
\end{lemma}

\begin{proof}[Proof of \cref{lem:rand_ps_algo_validity}]

As in the deterministic proof, the true language $K=L_{i^\star}$ is critical at all
sufficiently large time steps by \cref{lem:criticality}, and from some point onward the scope
size is always at least $i^\star$. Therefore, there is a time step $T^\star$ such that every focus after time step $T^\star$ is a subset of $K$. This time step may depend on the realization of the random permutations.

Every block created after time step $T^\star$ is contained in the current focus, and hence is
contained in $K$. Therefore, every output taken from a block created after time step $T^\star$ is in $K$. The
only possible outputs outside $K$ can come from blocks created up to time step $T^\star$, but
there are only finitely many such blocks and each has size $m$, so they can generate
only finitely many outputs. 

Finally, by construction, every output is chosen among numbers that have not been
previously announced by either party. This proves the lemma.
\end{proof}

Throughout the rest of this subsection, we fix an oblivious adversary. All probabilities and expectations are taken over the random permutations of the algorithm. The algorithm creates infinitely many blocks $B_1, B_2, \ldots$, indexed in the order of creation; let $t_r$ be the time step at which $B_r$ is created. We may assume that the permutation $\pi_r$ is drawn from an independent random seed $\xi_r$. Since the adversary's announcements do not depend on the algorithm's random choices, everything that happens before $\pi_r$ is first used is determined by $\xi_1, \ldots, \xi_{r-1}$. Let $\mathcal{H}_r$ denote the $\sigma$-algebra generated by $\xi_1, \ldots, \xi_{r-1}$, which represents the history at the creation of $B_r$. Conditional on $\mathcal H_r$, $\pi_r$ is uniformly distributed over the permutations of $B_r$.

As in the proof of \cref{thm:ps_algo_lower_density}, we may now identify $K$ with $\mathbb{Z}^+$ in an order-preserving way, and assume throughout the rest of this section that $K=\mathbb{Z}^+$. In particular, $[n]$ denotes the set of the $n$ smallest integers of $K$. Integers outside $K$ can be reserved only in the finitely many blocks created up to time step $T^\star$.

Every reserved integer belongs to a unique block. Moreover, whenever the adversary announces an unreserved integer, the algorithm immediately reserves a new block. We associate that integer with the new block and call it the \emph{trigger} of the block. For each block $B_r$, let $C_r$ denote the set consisting of the integers of $B_r$ and the trigger of $B_r$, if it has one. The sets $C_r$ are pairwise disjoint and cover $K$, and therefore
\begin{equation}\label{eq:rand_ps_algo_partition}
\sum_{r}\abs{C_r\cap[n]}=n \qquad\text{for every positive integer } n.
\end{equation} We classify the relevant blocks into two types.

\begin{definition}[Complete Blocks]
    A block $B_r$ is \emph{complete within} $n$ if $B_r \subseteq [n]$. Otherwise, we call $B_r$ \emph{incomplete within} $n$.
\end{definition}

Whether $B_r$ is complete within $n$ is determined at the creation of $B_r$, so the event $\set{B_r\subseteq[n]}$ is $\mathcal{H}_r$-measurable.

Each integer in $[n]$ falls into one of the following two cases:
\begin{enumerate}
    \item it belongs to $C_r$ for a block $B_r$ complete within $n$, that is, it belongs to such a block or is the trigger of such a block;
    \item it belongs to $C_r$ for a block $B_r$ incomplete within $n$.
\end{enumerate}

The rest of the proof can be summarized as follows: (1) in expectation, the algorithm obtains a constant fraction of the integers in the first case; and (2) for every realization, the number of integers falling into the second case is negligible compared with $n$. Therefore, even if the algorithm gives up on all integers in the second case, it can still win a constant fraction of all integers in expectation. This argument underlies the analysis of not only the \nameref{box:rand_ps_algo}, but also the other block-based algorithms in later sections.

We first compute the algorithm's density among integers in the first case.

\begin{lemma}\label{lem:rand_ps_algo_cmp_blks}
For any positive integer $n$ and any block $B_r$, on the event that $B_r$ is complete within $n$,
\[
\E{}{\abs{D \cap B_r} \mid \mathcal{H}_r} \geq \left(1 - \frac{1}{m + 1}\right)\cdot f(m + 1) \cdot \abs{C_r \cap [n]}.
\]
\end{lemma}
\begin{proof}[Proof of \cref{lem:rand_ps_algo_cmp_blks}]
Condition on $\mathcal{H}_r$ and suppose $B_r\subseteq[n]$. Let $y_1,\ldots,y_m$ be the elements of $B_r$ in the adversary's announcement order. Consider a bipartite graph with offline vertices $B_r$ and online vertices $u_0,u_1,\ldots,u_m$. Vertex $u_0$ represents the creation step $t_r$ and is adjacent to all of $B_r$; for each $j\in[m]$, vertex $u_j$ represents the adversary's announcement of $y_j$ and has neighbors $\{y_{j+1},\ldots,y_m\}$. This construction applies whether or not $B_r$ has a trigger. The graph is determined by $\mathcal{H}_r$ and has a matching of size $m$, obtained by matching $u_{j-1}$ to $y_j$ for each $j\in[m]$. The algorithm's announcements within $B_r$ correspond exactly to the matches produced by Ranking with permutation $\pi_r$. Conditional on $\mathcal{H}_r$, $\pi_r$ is uniformly distributed over the permutations of $B_r$. Thus, \cref{thm:ranking_algo} yields
\[
\E{}{\abs{D \cap B_r} \mid \mathcal{H}_r} \geq m \cdot f(m + 1).
\]
Finally, by $\abs{C_r \cap [n]} \leq \abs{C_r} \leq m + 1$, we have $m \cdot f(m + 1) \geq \left(1 - \frac1{m + 1}\right) \cdot f(m + 1) \cdot \abs{C_r \cap [n]}$.
\end{proof}

Next, we count the number of blocks involved in the second case.

\begin{lemma}\label{lem:rand_ps_algo_incmp_blks} For every realization and any positive integer $n \geq m$, there are at most $\log_2\frac{n}{m} + 1$ blocks $B_r$ created after time step $T^\star$ that are incomplete within $n$ and satisfy either $B_r\cap[n]\neq\emptyset$ or $B_r$ has a trigger at most $n$.
\end{lemma}

\begin{proof}[Proof of \cref{lem:rand_ps_algo_incmp_blks}]
    Once an incomplete block within $n$ has been reserved, during any subsequent interval with no switch loss, every newly reserved block, together with its trigger if it has one, lies above $n$. Hence, later newly reserved blocks are not counted by the lemma until a switch loss $\ell$ either is an unreserved integer at most $n$ that triggers the new block, or makes the smallest integer in $W \cap F$ at most $n$.

    Consider such a switch loss $\ell$, caused by a focus change from $L_j$ to $L_i$. Then $i < j$ and $L_j \subseteq L_i \subseteq K$. During the progression of the focus from $L_i$ to $L_j$, the \nameref{box:rand_ps_algo} must have reserved $2^{\tau^\prime}$ blocks while its focus was $L_i$, where $\tau^\prime$ is the value of $\tau$ during that period. If the smallest integer in $W \cap F$ after switching the focus back to $L_i$ is at most $n$, then all these blocks lie entirely below that integer, and hence below $n$. If instead $\ell$ is an unreserved integer at most $n$, then all these blocks lie entirely below $\ell$, and hence below $n$.

    We now use the same charging argument as in \cref{lem:ps_algo_switch_loss}: charge the switch loss $\ell$ to these $2^{\tau^\prime}$ blocks. If there are $w$ such switch losses $\ell_1, \ell_2, \ldots, \ell_w$ with corresponding focus-change counts $\tau_1, \tau_2, \ldots, \tau_w$, then disjointness of the charged sets implies
    \[
    n \geq \left(2^{\tau_1} + 2^{\tau_2} + \cdots + 2^{\tau_w}\right) \cdot m.
    \]
    Since $\tau_1, \tau_2, \ldots, \tau_w$ are distinct positive integers, we have $n \geq \left(2^1+\cdots+2^w\right)\cdot m \geq 2^w \cdot m$. Taking logarithms gives $w \leq \log_2\frac{n}{m}$. The first incomplete block contributes the additional $1$.
\end{proof}

Finally, we can compute the lower density of the \nameref{box:rand_ps_algo}.

\begin{proof}[Proof of \cref{thm:rand_ps_algo_lower_density_weaker}]
   Fix an arbitrary prefix $[n]$, and write $\phi \defeq \left(1 - \frac1{m+1}\right) \cdot f(m + 1)$. For each block $B_r$, let $I_r \defeq \mathbf{1}\set{B_r \subseteq [n]}$ indicate that $B_r$ is complete within $n$; it is $\mathcal{H}_r$-measurable. Let
   \[
   \Delta_n \defeq \sum_r (1 - I_r) \cdot \abs{C_r \cap [n]}
   \]
   be the number of integers in the second case. By \eqref{eq:rand_ps_algo_partition}, $\sum_r I_r \cdot \abs{C_r \cap [n]} = n - \Delta_n$.

   \begin{enumerate}
    \item For the integers in the first case: since the blocks are pairwise disjoint, we have $\mu_n(D \cap K) \geq \sum_r I_r \cdot \abs{D \cap B_r}$. Since $I_r$ and $\abs{C_r \cap [n]}$ are $\mathcal{H}_r$-measurable, the law of total expectation and \cref{lem:rand_ps_algo_cmp_blks} give, for every $r$,
    \[
    \E{}{I_r \cdot \abs{D \cap B_r}} = \E{}{I_r \cdot \E{}{\abs{D \cap B_r} \mid \mathcal{H}_r}} \geq \phi \cdot \E{}{I_r \cdot \abs{C_r \cap [n]}}.
    \]
    Summing over $r$, and using monotone convergence for these series of nonnegative terms, we obtain
    \begin{equation}\label{eq:rand_ps_algo_first_case}
    \E{}{\mu_n(D \cap K)} \geq \phi \cdot \E{}{n - \Delta_n}.
    \end{equation}
    \item For the integers in the second case, fix a realization, and let $M_0$ be the total number of integers in the sets $C_r$ of the blocks created up to time step $T^\star$. By \cref{lem:rand_ps_algo_validity}, $M_0$ is finite and does not depend on $n$. Every block created after time step $T^\star$ contributes to $\Delta_n$ only if it is incomplete within $n$ and $C_r \cap [n] \neq \emptyset$. By \cref{lem:rand_ps_algo_incmp_blks}, there are at most $\log_2\frac{n}{m} + 1$ such blocks, each contributing at most $m + 1$ integers. Therefore,
    \[
    \Delta_n \leq M_0 + (m + 1) \cdot \left(\log_2\frac{n}{m} + 1\right),
    \]
    and hence $(n - \Delta_n) / n \to 1$ as $n \to \infty$, for every realization.
   \end{enumerate}

   Note that $T^\star$, and hence $M_0$, may depend on the algorithm's random choices, so we pass to the limit using Fatou's lemma. Since $\mu_n(K) = n$ and $(n - \Delta_n)/n \geq 0$, \eqref{eq:rand_ps_algo_first_case} and Fatou's lemma yield that the algorithm's lower density is at least
   \begin{align*}
    \liminf_{n \to \infty}\frac{\E{}{\mu_n(D\cap K)}}{\mu_n(K)}&\geq\phi \cdot \liminf_{n \to \infty}\E{}{\frac{n - \Delta_n}{n}} \geq \phi \cdot \E{}{\liminf_{n \to \infty}\frac{n - \Delta_n}{n}}\\
    &=\left(1 - \frac1{m+1}\right) \cdot f(m+1). \qedhere
   \end{align*}
\end{proof}

\subsection{Upper Bound}\label{subsec:rand_ps_algo_lower_density_upper}
In this section, we present an example of a single-language game $\mathcal{L} = \set{\Z^+}$ in which no randomized algorithm can have lower density greater than $1 - 1/e$.

Partition $\Z^+$ into consecutive blocks $B_1, B_2, \ldots$ of doubly exponentially increasing size, where block $B_i$ has size $2^{2^i}$. Let $n_i = \sum_{j = 1}^i\abs{B_j}$. The adversary will enumerate $\Z^+$ block by block.

\begin{lemma}\label{lem:rand_algo_hard_per_blk}
Let $B$ be a block of size $m$, and suppose that a randomized algorithm has already announced $o(m)$ integers before the adversary starts announcing any element of $B$. Then there is a permutation $\pi$ of $B$ such that, when the adversary announces the elements of $B$ in the order $\pi$, the algorithm's density in $B$ is at most
\[
1 - \frac1e + o_m(1).
\]
\end{lemma}

Using \cref{lem:rand_algo_hard_per_blk}, we can construct, for each randomized algorithm, a hard permutation of $\Z^+$.

\begin{proof}[Proof of \cref{thm:rand_ps_algo_lower_density_upper} assuming \cref{lem:rand_algo_hard_per_blk}] Fix a randomized algorithm. We construct the permutation inductively. Assume that we have computed the restriction of $\pi$ to $B_1, \dots, B_{i - 1}$.
    
Consider the randomized algorithm obtained by feeding the adversary's input on the first $i - 1$ blocks to the algorithm. Since the adversary takes $n_{i - 1}$ time steps to enumerate the first $i - 1$ blocks, the algorithm can announce at most this many integers before the adversary starts enumerating $B_i$. Moreover, it holds that
\[
\frac{n_{i - 1}}{\abs{B_i}} \to 0
\qquad\text{as }i\to\infty.
\]
Hence, by \cref{lem:rand_algo_hard_per_blk}, there is a permutation $\pi_i$ of $B_i$ such that the algorithm's density in $B_i$ is at most $1 - 1/e + o_{\abs{B_i}}(1)$, and hence at most $1 - 1 / e + o_{n_i}(1)$. It follows that the algorithm's density in the first $i$ blocks is at most
\[
\frac{n_{i - 1} + \E{}{\abs{B_i \cap D}}}{n_i}
\leq  \frac{n_{i - 1}}{n_i} + \frac{\E{}{\abs{B_i \cap D}}}{\abs{B_i}}
\leq \frac{n_{i - 1}}{n_i} + 1 - \frac1e + o_{n_i}(1) \leq 1 - \frac1e + o_{n_i}(1).
\]
We then extend $\pi$ to the first $i$ blocks by appending $\pi_i$.

Given this construction, for infinitely many integers $n_1, n_2, \ldots$, we have
\[
\frac{\E{}{\mu_{n_i}(D \cap K)}}{\mu_{n_i}(K)} \leq\frac{n_{i - 1} + \E{}{\abs{B_i \cap D}}}{n_i} \leq 1 - \frac1e + o_{n_i}(1).
\]
Therefore, the lower density cannot be greater than $1 - 1/e$.
\end{proof}

\begin{remark}
    Since the adversary's announcement order $\pi$ is constructed inductively, and later blocks are chosen after considering the algorithm's behavior on earlier blocks, the strategy may appear to be adaptive. This is not the case: after the randomized algorithm $A$ is fixed, the construction fixes a single order $\pi$ by reasoning about the distribution of $A$, and this order does not depend on the realized random choices of $A$ during any execution.
\end{remark}

\begin{proof}[Proof of \cref{lem:rand_algo_hard_per_blk}]
Fix the randomized algorithm and the adversary's input before the adversary starts announcing elements of $B$. Let $P \subseteq B$ be the random set of elements of $B$ that the algorithm has already announced. By assumption, $\abs{P} = o(m)$ for every realization.

For every permutation $\sigma$ of $B$, define an online bipartite matching instance $G_\sigma = (U \sqcup B, E)$ as follows. The online vertices are $u_1,\ldots,u_m$, and $u_t$ is adjacent to the elements of $B$ that appear after $\sigma(t)$ in the order $\sigma$. Equivalently,
\[
N(u_t)=\{\sigma(t + 1),\sigma(t + 2),\ldots,\sigma(m)\}.
\]
The given randomized algorithm induces a randomized online matching algorithm on these instances: when $u_t$ arrives, simulate the adversary announcing $\sigma(t)$ to the original algorithm; if the algorithm outputs an unmatched neighbor of $u_t$, match $u_t$ to this vertex, and otherwise leave $u_t$ unmatched. Thus the number of elements of $B$ first announced by the original algorithm during the block is at most the size of the matching produced in $G_\sigma$.

By \cref{thm:online_matching_upper}, there is a strictly upper-triangular hard instance for this randomized matching algorithm on which the expected matching size is at most
\[
\left(1 - \frac1e + o_m(1)\right) \cdot m.
\]
Since every such strictly upper-triangular instance is $G_\pi$ for some permutation $\pi$ of $B$, this permutation satisfies
\[
\E{}{\abs{B \cap D}} \leq \E{}{\abs{P}} + \left(1 - \frac1e + o_m(1)\right) \cdot m
\leq \left(1 - \frac1e + o_m(1)\right) \cdot m.
\]
Therefore, the algorithm's density in $B$ is at most $1 - 1/e + o_m(1)$ in expectation.
\end{proof}

\begin{remark}
If the algorithm is allowed to announce $q$ integers at each time step, the same randomized strategy achieves optimal lower density
\[
q \cdot \left(1-\exp(-1/q)\right)=1-O(1/q).
\]
At each step, the algorithm outputs the first $q$ currently unannounced integers according to the relevant block's random permutation. The stated ratio follows from a differential-equation analysis analogous to that of the Ranking algorithm by \cite{DBLP:conf/stoc/KarpVV90}.
\end{remark}

\section{Language Generation under Multiple Orders}\label{sec:paired_ps_algo}
The previous definition of density depends on the order of the positive integers. In this section, we extend the language generation problem to a multi-order setting. Given $k$ orders $\sigma_1, \ldots, \sigma_k$, can any deterministic algorithm simultaneously achieve good lower densities for all $k$ orders against any adaptive adversary?

Formally, let $\sigma_1, \ldots, \sigma_k$ be permutations of $\Z^+$; that is, each $\sigma_i\colon \Z^+ \to \Z^+$ is a bijection. Let
\[
\mu_{i, n}(S) \defeq \bigl\lvert[n]_i \cap S\bigr\rvert,
\] the number of integers in $S$ that are among the first $n$ elements in the order $\sigma_i$.

\begin{definition}[Lower Density under $k$ Orders] The lower density of an algorithm under $k$ orders $\sigma_1, \ldots, \sigma_k$ is
    \[
    \min_{i \in [k]}\liminf_{n \to \infty}\frac{\mu_{i, n}(D \cap K)}{\mu_{i, n}(K)}.
    \]
\end{definition}

An immediate upper bound is $1/2$, since the optimal lower density under a single order is $1/2$. In this section, we show that, surprisingly, having multiple orders does not decrease the best achievable lower density.

\begin{theorem}[Lower Density under $k$ Orders]\label{thm:paired_ps_algo_lower_density} There is a deterministic algorithm that, against every adaptive adversary, achieves lower density $1/2$ under $k$ orders.
\end{theorem}

To streamline the presentation, we defer the proof of \cref{thm:paired_ps_algo_lower_density} to \cref{apx:paired_ps_algo_lower_density}. In this section, we instead prove a slightly weaker statement that preserves the main idea of the algorithm.

\begin{theorem}[Lower Density under $k$ Orders (Weaker Lower Bound)]\label{thm:paired_ps_algo_lower_density_weaker}
For every positive integer $m$, there is an algorithm that, against every adaptive adversary, achieves lower density
\[
\frac12 h(m),
\]
where
\[
h(m) = 1 - \frac{4\cdot(km+1)^{1 - 1/k} + 1}{m}.
\]
This lower density tends to $1/2$ as $m \to \infty$.
\end{theorem}

\subsection{Warm-Up}

We illustrate the idea by restricting the game to a finite language $L = [m]$, where $m$ is even.

Apply \cref{thm:low_crossing} to partition $[m]$ into $m / 2$ pairs; see \cref{fig:low-crossing} for example. Define $P(x)$ to be the other integer in the pair containing $x$. The algorithm is simple: whenever the adversary announces an integer $x$, the algorithm replies with $P(x)$, if it has not yet been announced.

\begin{figure}[t!]\centering
\begin{tikzpicture}[
    every node/.style={font=\scriptsize},
    num/.style={
        circle,
        draw=ForestGreen!80!black,
        ForestGreen!80!black,
        fill=ForestGreen!15,
        inner sep=1.2pt,
        minimum size=17pt
    },
    uncut/.style={draw=black!35, line width=0.9pt},
    cutpair/.style={draw=red!80!black, line width=1.4pt},
    cutline/.style={dotted, very thick, draw=orange!65!black}
]

\def\cutx{5.5}

\fill[orange!30] (-0.55,-0.55) rectangle (\cutx,0.55);
\fill[orange!30] (-0.55,-2.85) rectangle (\cutx,-1.75);

\node[anchor=east] at (-1.0,0) {$\sigma_1$};
\foreach \a [count=\i from 0] in {1,2,3,4,5,6,7,8,9,10,11,12} {
    \node[num] (t\a) at (\i,0) {\a};
}

\node[anchor=east] at (-1.0,-2.3) {$\sigma_2$};
\foreach \a [count=\i from 0] in {1,5,7,2,11,9,3,6,8,4,12,10} {
    \node[num] (b\a) at (\i,-2.3) {\a};
}

\draw[cutline] (\cutx,-0.85) -- (\cutx,2.35);
\node[orange!65!black, anchor=south] at (\cutx,2.35) {cut};

\draw[cutline] (\cutx,-4.75) -- (\cutx,-1.45);
\node[orange!65!black, anchor=north] at (\cutx,-4.75) {cut};

\draw[uncut]   (t1.north)  .. controls +(0,0.55) and +(0,0.55) .. (t2.north);
\draw[uncut]   (t3.north)  .. controls +(0,0.55) and +(0,0.55) .. (t4.north);
\draw[cutpair] (t5.north)  .. controls +(0,1.45) and +(0,1.45) .. (t11.north);
\draw[cutpair] (t6.north)  .. controls +(0,1.95) and +(0,1.95) .. (t12.north);
\draw[uncut]   (t7.north)  .. controls +(0,0.55) and +(0,0.55) .. (t8.north);
\draw[uncut]   (t9.north)  .. controls +(0,0.55) and +(0,0.55) .. (t10.north);

\draw[uncut]   (b1.south)  .. controls +(0,-0.85) and +(0,-0.85) .. (b2.south);
\draw[uncut]   (b3.south)  .. controls +(0,-0.85) and +(0,-0.85) .. (b4.south);
\draw[uncut]   (b5.south)  .. controls +(0,-0.65) and +(0,-0.65) .. (b11.south);
\draw[uncut]   (b6.south)  .. controls +(0,-0.65) and +(0,-0.65) .. (b12.south);
\draw[cutpair] (b7.south)  .. controls +(0,-1.45) and +(0,-1.45) .. (b8.south);
\draw[cutpair] (b9.south)  .. controls +(0,-1.95) and +(0,-1.95) .. (b10.south);

\draw[uncut] (0,-6.0) -- +(0.9,0);
\node[anchor=west] at (1.1,-6.0) {uncrossed pair};

\draw[cutpair] (4.0,-6.0) -- +(0.9,0);
\node[anchor=west] at (5.1,-6.0) {pair crossed by the displayed cut};

\end{tikzpicture}
\caption{An example of a low-crossing pairing for two orders $\sigma_1$ and $\sigma_2$.}
\label{fig:low-crossing}
\end{figure}
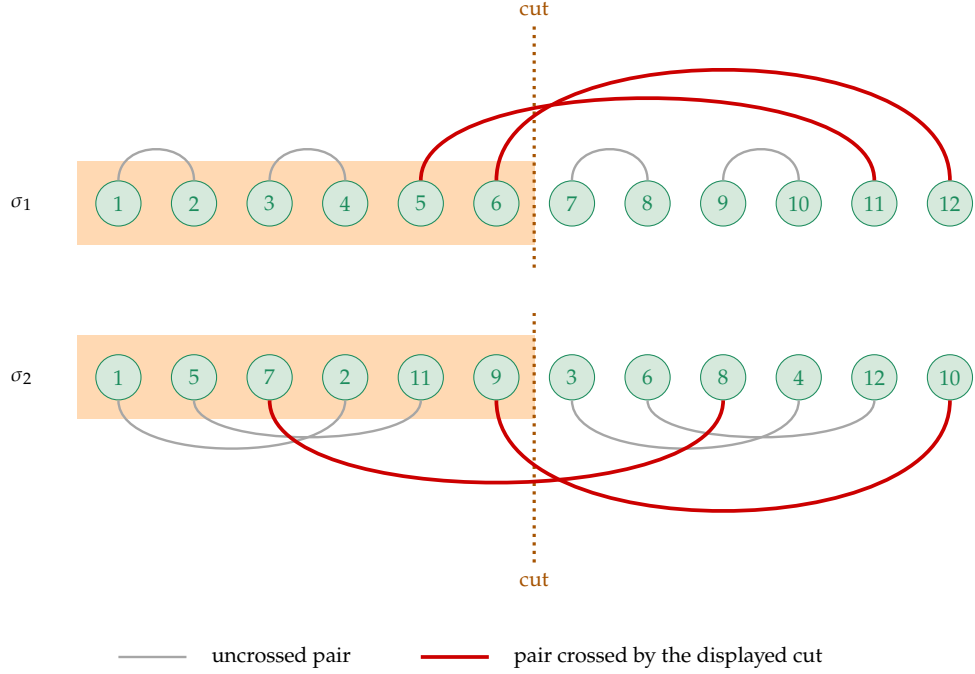

We now compute the density of each prefix $[n]_i$. For each pair in the partition of $[m]$, exactly one of the two elements is first announced by the algorithm. Therefore, it suffices to count the number of pairs crossed by the cut $([n]_i, [m]\setminus [n]_i)$, which is
\[
c_i(n) \leq 4 \cdot \left(\frac m2\right)^{1 - 1/k} \leq 4 \cdot m^{1 - 1/k}
\]
by \cref{thm:low_crossing}. Hence, in any sufficiently large prefix (say, $n > m/2$), the density of the algorithm is at least
\[\min_{i \in [k], n \in [m]\setminus[m/2]}\set{\frac{n - c_i(n)}{2n}} \geq \frac12 - 4 \cdot m^{-1/k}.\]

As in the randomized algorithm, we need to handle focus changes, so we reserve blocks on the fly similarly. The \nameref{box:rand_ps_algo} always reserves the smallest $m$ numbers in the current focus that have not been announced by either party or reserved by the algorithm. With $k$ orders, however, the algorithm must take all of these orders into account, which makes the situation more complicated.

\subsection{The Paired Patient-Scope Algorithm}

As in the \nameref{box:rand_ps_algo}, we fix a positive integer $m$,
which serves as a block-size parameter. The algorithm maintains a family of
pairwise disjoint blocks $B_1, B_2, \dots$, each of size between $m$ and $km+1$.
Each block $B_r$ is partitioned into pairs (denoted by $P_r$) by applying
\cref{thm:low_crossing} to the orders induced by
$\sigma_1, \dots, \sigma_k$ on $B_r$.

\begin{mybox}[label={box:paired_ps_algo},nameref={paired patient-scope algorithm}]
{The paired patient-scope algorithm}
\begin{itemize}
    \item Fix a block size $m\in\mathbb{Z}^+$. Let $N_B=0$ be the current number of blocks.
    \item Initially, at time $t=0$, set the scope size $s_0=1$, set the focus-change
    count $\tau=1$, and create no blocks.
    \item In each time step $t$, perform the following operations.
    \begin{itemize}
        \item Set the scope size $s_t = s_{t - 1}$.
        \item Receive a new integer from the adversary, and update the sets of consistent and critical languages accordingly.
        \item If the focus becomes inconsistent, run the \nameref{box:bt_algo}.
        \item Otherwise, if the focus $F$ is still consistent and the algorithm has reserved $2^\tau$ blocks since the last focus change, increase the scope size by one and update
        the focus accordingly. If the focus changes, increase $\tau$ by one.
        \item Let $x_t$ be the adversary's announced number in this step.
        \begin{itemize}
            \item If $x_t\in B_r$ for some block $B_r$ and $P_r(x_t)$ has not yet been announced
            by either party, output $P_r(x_t)$.
            \item Otherwise, let
            \[
            B_{N_B+1} \defeq \bigcup_{i \in [k]} (W \cap F)[m]_i.
            \] Namely, block $B_{N_B+1}$ contains, for each order, the $m$ smallest available elements of the current focus $F$ under that order. If $\abs{B_{N_B + 1}}$ is odd, we further add the smallest element in $(W \cap F)\setminus B_{N_B + 1}$ under $\sigma_1$ to $B_{N_B + 1}$. Partition $B_{N_B+1}$ into pairs $P_{N_B+1}$ by applying
            \cref{thm:low_crossing} to the induced orders of
            $\sigma_1, \dots, \sigma_k$ on $B_{N_B+1}$. Mark all elements of $B_{N_B+1}$ as
            reserved, and output the first element of $B_{N_B+1}$ under $\sigma_1$. Increase $N_B$
            by $1$.
        \end{itemize}
    \end{itemize}
\end{itemize}
\end{mybox}

We first show the validity of \nameref{box:paired_ps_algo}.

\begin{lemma}[Validity]\label{lem:paired_ps_algo_validity} The \nameref{box:paired_ps_algo} can generate in the limit the true language $K$.
\end{lemma}

\begin{proof}[Proof of \cref{lem:paired_ps_algo_validity}]

As in the deterministic proof, the true language $K=L_{i^\star}$ is critical at all
sufficiently large time steps by \cref{lem:criticality}, and from some point onward the scope
size is always at least $i^\star$. Therefore, there is a time step $T^\star$ such that every focus after time step $T^\star$ is a subset of $K$.

Every block created after time step $T^\star$ is contained in the current focus, and hence is
contained in $K$. Therefore, every output taken from a block created after time step $T^\star$ is in $K$. The
only possible outputs outside $K$ can come from blocks created up to time step $T^\star$, but
there are only finitely many such blocks and each has size at most $km + 1$, so they can generate
only finitely many outputs.

Finally, by construction, every output is chosen among numbers that have not been
previously announced by either party. This proves the lemma.
\end{proof}

As in the proofs of \cref{thm:ps_algo_lower_density,thm:rand_ps_algo_lower_density_weaker}, we may now identify $K$ with $\Z^+$ in an order-preserving way, update each $\sigma_i$ accordingly, and assume throughout the rest of this section that $K = \Z^+$. In particular, $[n]_i$ denotes the set of the first $n$ integers of $K$ under $\sigma_i$. Integers outside $K$ can be reserved only in the finitely many blocks created up to time step $T^\star$.

We define the trigger of a block and the sets $C_r$ in the same way as in \cref{sec:rand_ps_algo}, and therefore
\begin{equation}\label{eq:paired_ps_algo_partition}
\sum_{r}\abs{C_r\cap[n]_i}=n \qquad\text{for every } i \in [k] \text{ and every positive integer } n.
\end{equation}

\begin{figure}[t!]\centering
\begin{tikzpicture}[x=0.58cm,y=0.72cm,font=\small]
  \begin{scope}[shift={(0,8.0)}]
    \node[font=\bfseries,align=center] at (5.6,4.35)
      {The randomized patient-scope algorithm};
    \node[align=center] at (5.6,3.55)
      {whole blocks are monotone under $\mathbf{id}$};

    \draw[->,thick] (0,0) -- (12.6,0) node[right] {$\mathbf{id}$-order};
    \foreach \x in {1,...,12} {
      \draw (\x,0.12) -- (\x,-0.12);
    }

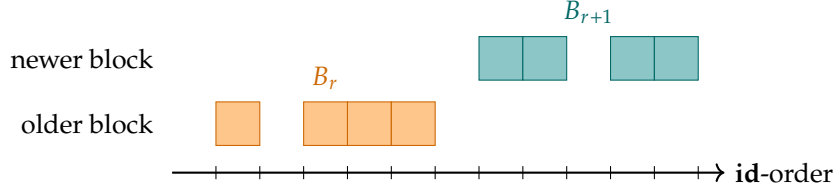
    \node[left] at (-0.2,2.1) {newer block};
    \node[left] at (-0.2,0.9) {older block};

    \foreach \x in {1,3,4,5} {
      \filldraw[fill=orange!45,draw=orange!80!black] (\x,0.5) rectangle +(1,0.8);
    }
    \foreach \x in {7,8,10,11} {
      \filldraw[fill=teal!45,draw=teal!80!black] (\x,1.7) rectangle +(1,0.8);
    }

    \node[orange!80!black] at (3.5,1.75) {$B_r$};
    \node[teal!80!black] at (9.5,2.95) {$B_{r+1}$};

    \node[align=center,text width=8.8cm,font=\footnotesize] at (5.6,-1.45)
      {During a stretch with no switch loss, each new block lies entirely above the older blocks under the only order $\mathbf{id}$.};
  \end{scope}

  \begin{scope}[shift={(0,0)}]
    \node[font=\bfseries,align=center] at (5.6,4.35)
      {The paired patient-scope algorithm};
    \node[align=center] at (5.6,3.55)
      {the whole blocks may interleave, but the heads remain monotone};

    \draw[->,thick] (0,0) -- (12.6,0) node[right] {$\sigma_i$-order};
    \foreach \x in {1,...,12} {
      \draw (\x,0.12) -- (\x,-0.12);
    }

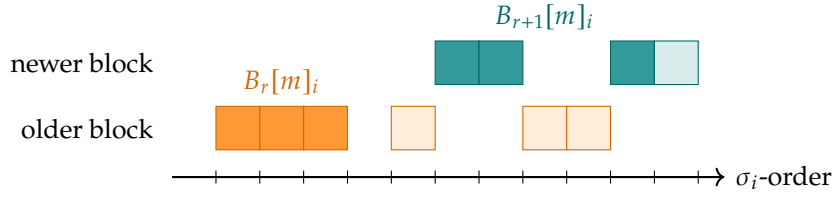
    \node[left] at (-0.2,2.1) {newer block};
    \node[left] at (-0.2,0.9) {older block};

    \foreach \x in {1,2,3} {
      \filldraw[fill=orange!80,draw=orange!80!black] (\x,0.5) rectangle +(1,0.8);
    }
    \foreach \x in {5,8,9} {
      \filldraw[fill=orange!15,draw=orange!80!black] (\x,0.5) rectangle +(1,0.8);
    }

    \foreach \x in {6,7, 10} {
      \filldraw[fill=teal!80,draw=teal!80!black] (\x,1.7) rectangle +(1,0.8);
    }
    \foreach \x in {11} {
      \filldraw[fill=teal!15,draw=teal!80!black] (\x,1.7) rectangle +(1,0.8);
    }

    \node[orange!80!black] at (2.5,1.75) {$B_r[m]_i$};
    \node[teal!80!black] at (8.5,2.95) {$B_{r+1}[m]_i$};

    \node[align=center,text width=8.8cm,font=\footnotesize] at (5.6,-1.9)
      {Dark cells form the heads under $\sigma_i$. Light cells are additional elements selected because of other orders. The full blocks need not be monotone under $\sigma_i$, but the heads still move upward monotonically.};
  \end{scope}
\end{tikzpicture}
\caption{A schematic comparison under a fixed order $\sigma_i$ during a stretch with no switch loss. In the \nameref{box:rand_ps_algo}, whole blocks are monotone under $\mathbf{id}$. In the \nameref{box:paired_ps_algo}, a block may also contain elements chosen from $(W \cap F)[m]_j$ for $j \neq i$, so the full blocks can interleave under $\sigma_i$; however, the heads $B[m]_i$ remain monotone.}
\label{fig:monotone-blocks}
\end{figure}

A natural first attempt is to classify, for each prefix $[n]_i$, the relevant blocks into complete and incomplete blocks within $[n]_i$, as in the proof of \cref{thm:rand_ps_algo_lower_density_weaker}. However, this approach no longer works under multiple orders. In the \nameref{box:rand_ps_algo}, blocks created during a stretch with no switch loss are positioned monotonically under the single order. In the \nameref{box:paired_ps_algo}, this is no longer literally true: a block may contain elements from $(W \cap F)[m]_j$ with $j \neq i$, which can create gaps under the
order $\sigma_i$. To capture the monotone structure that remains, we isolate the
first $m$ elements of each block under a fixed order.

For a block $B_r$, define its \emph{head under $\sigma_i$} to be the set
$B_r[m]_i$. During any consecutive interval with no switch loss, the heads of
newly reserved blocks are strictly above the heads of older reserved blocks
under $\sigma_i$ (see \cref{fig:monotone-blocks} for comparison).

\begin{definition}[Complete head within $n$]
A block head $B_r[m]_i$ is \emph{complete within $n$} if
\[
B_r[m]_i \subseteq [n]_i.
\] Otherwise, we call the head $B_r[m]_i$ \emph{incomplete within $n$}.
\end{definition}

Each integer in $[n]_i$ falls into one of the following two cases:
\begin{enumerate}
    \item it belongs to $C_r$ for a block $B_r\subseteq K$ whose head $B_r[m]_i$ is complete within $n$, that is, it belongs to such a block or is the trigger of such a block;
    \item it belongs to $C_r$ for any other block $B_r$.
\end{enumerate}

We first compute the algorithm's density among integers in the first case.

\begin{lemma}\label{lem:paired_ps_algo_cmp_heads}
    For any order index $i$, any positive integer $n$, and any block $B_r\subseteq K$ whose head $B_r[m]_i$ is complete within $n$,
\[
\abs{D \cap B_r \cap [n]_i} \geq \frac12 h(m) \cdot \abs{C_r \cap [n]_i},
\] where
\[
h(m) = 1 - \frac{4\cdot(km+1)^{1 - 1/k} + 1}{m}.
\]
\end{lemma}
\begin{proof}[Proof of \cref{lem:paired_ps_algo_cmp_heads}]
Every pair of $P_r$ contained in $[n]_i$ contains an integer of $D$, since whenever the adversary first announces one integer of the pair, the algorithm announces the other one if it has not yet been announced.

Moreover, since $B_r \subseteq K$, the set $B_r \cap [n]_i$ is a prefix of $B_r$ under the order induced by $\sigma_i$. By \cref{thm:low_crossing}, applied to the induced orders on $B_r$, whose size is at most $km + 1$, the number $c$ of pairs of $P_r$ crossed by the cut $(B_r \cap [n]_i, B_r \setminus [n]_i)$ satisfies
\[
c \leq 4\cdot \left(\frac{km + 1}2\right)^{1 - 1/k} \leq 4 \cdot (km + 1)^{1 - 1/k}.
\]
Each crossed pair has exactly one integer in $[n]_i$, and every other pair is either contained in $[n]_i$ or disjoint from it. Hence, exactly $(\abs{B_r \cap [n]_i} - c)/2$ pairs are contained in $[n]_i$, and each of them contributes an integer to $D \cap B_r \cap [n]_i$. Since $C_r$ consists of $B_r$ and at most one trigger, we obtain
\[
\abs{D \cap B_r \cap [n]_i} \geq \frac{\abs{B_r \cap [n]_i} - c}{2} \geq \frac{\abs{C_r \cap [n]_i} - c - 1}{2}.
\]
Finally, since the head is complete within $n$, we have $B_r[m]_i \subseteq [n]_i$, and hence $\abs{C_r \cap [n]_i} \geq m$. Therefore, $c + 1 \leq (c + 1) \cdot \abs{C_r \cap [n]_i} / m$, and
\[
\abs{D \cap B_r \cap [n]_i} \geq \frac12 \cdot \left(1 - \frac{c + 1}{m}\right) \cdot \abs{C_r \cap [n]_i} \geq \frac12 h(m) \cdot \abs{C_r \cap [n]_i}.\qedhere
\]
\end{proof}

Next, we count the number of blocks involved in the second case.

\begin{lemma}\label{lem:paired_ps_algo_incmp_heads} For any order index $i$ and any positive integer $n \geq m$, there are at most $\log_2\frac{n}{m} + 1$ blocks $B_r$ created after time step $T^\star$ whose heads $B_r[m]_i$ are incomplete within $n$ and that satisfy either $B_r\cap[n]_i\neq\emptyset$ or $B_r$ has a trigger in $[n]_i$.
\end{lemma}

\begin{proof}[Proof of \cref{lem:paired_ps_algo_incmp_heads}]
    Once a block with an incomplete head within $n$ has been reserved, during any subsequent interval with no switch loss, every newly reserved block, together with its trigger if it has one, lies outside $[n]_i$. Hence, later newly reserved blocks are not counted by the lemma until a switch loss $\ell$ either is an unreserved integer in $[n]_i$ that triggers the new block, or makes the smallest integer in $W \cap F$ under $\sigma_i$ lie in $[n]_i$.

    Consider such a switch loss $\ell$, caused by a focus change from $L_q$ to $L_p$. Then $p < q$ and $L_q \subseteq L_p \subseteq K$. During the progression of the focus from $L_p$ to $L_q$, the \nameref{box:paired_ps_algo} must have reserved $2^{\tau^\prime}$ blocks while its focus was $L_p$, where $\tau^\prime$ is the value of $\tau$ during that period. If, after the switch back to $L_p$, the smallest integer in $W \cap F$ under $\sigma_i$ lies in $[n]_i$, then the heads of all these blocks lie before that element under $\sigma_i$, and hence lie entirely in $[n]_i$. If instead $\ell$ is an unreserved integer of $[n]_i$, then these heads lie before $\ell$ under $\sigma_i$, and hence lie entirely in $[n]_i$.

    We now use the same charging argument as in \cref{lem:ps_algo_switch_loss,lem:rand_ps_algo_incmp_blks}: charge the switch loss $\ell$ to these $2^{\tau^\prime}$ heads. If there are $w$ such switch losses $\ell_1, \ell_2, \ldots, \ell_w$ with corresponding focus-change counts $\tau_1, \tau_2, \ldots, \tau_w$, then disjointness of the charged sets implies
    \[
    n \geq \left(2^{\tau_1} + 2^{\tau_2} + \cdots + 2^{\tau_w}\right) \cdot m.
    \]
    Since $\tau_1, \tau_2, \ldots, \tau_w$ are distinct positive integers, we have $n \geq \left(2^1+\cdots+2^w\right)\cdot m \geq 2^w \cdot m$. Taking logarithms gives $w \leq \log_2\frac{n}{m}$. The first block with an incomplete head contributes the additional $1$.
\end{proof}

Finally, we can compute the lower density of the \nameref{box:paired_ps_algo}.

\begin{proof}[Proof of \cref{thm:paired_ps_algo_lower_density_weaker}]
If $h(m)\leq 0$, the statement is trivial. Thus assume $h(m)>0$.

Fix an order index $i$ and an arbitrary prefix $[n]_i$. For each block $B_r$, let
\[
I_r \defeq \mathbf{1}\set{B_r \subseteq K \text{ and } B_r[m]_i \subseteq [n]_i}
\]
indicate that $B_r \subseteq K$ and the head $B_r[m]_i$ is complete within $n$. Let
\[
\Delta_{i,n} \defeq \sum_r (1 - I_r) \cdot \abs{C_r \cap [n]_i}
\]
be the number of integers in the second case. By \eqref{eq:paired_ps_algo_partition}, $\sum_r I_r \cdot \abs{C_r \cap [n]_i} = n - \Delta_{i,n}$.

\begin{enumerate}
    \item For the integers in the first case: since the blocks are pairwise disjoint, we have $\mu_{i,n}(D \cap K) \geq \sum_r I_r \cdot \abs{D \cap B_r \cap [n]_i}$. By \cref{lem:paired_ps_algo_cmp_heads}, we obtain
    \begin{equation}\label{eq:paired_ps_algo_first_case}
    \mu_{i,n}(D \cap K) \geq \frac12 h(m) \cdot \sum_r I_r \cdot \abs{C_r \cap [n]_i} = \frac12 h(m) \cdot \left(n - \Delta_{i,n}\right).
    \end{equation}
    \item For the integers in the second case, let $M_0$ be the total number of integers in the sets $C_r$ of the blocks created up to time step $T^\star$. By \cref{lem:paired_ps_algo_validity}, $M_0$ is finite and depends on neither $i$ nor $n$. Every block created after time step $T^\star$ is contained in $K$, and hence contributes to $\Delta_{i,n}$ only if its head $B_r[m]_i$ is incomplete within $n$ and $C_r \cap [n]_i \neq \emptyset$. By \cref{lem:paired_ps_algo_incmp_heads}, there are at most $\log_2\frac{n}{m} + 1$ such blocks. Each such block has size at most $km+1$ and has at most one trigger, so it contributes at most $km + 2$ integers. Therefore,
    \[
    \Delta_{i,n} \leq M_0 + (km+2)\cdot \left(\log_2\frac{n}{m}+1\right).
    \]
\end{enumerate}

Since $\mu_{i,n}(K) = n$, \eqref{eq:paired_ps_algo_first_case} and the bound on $\Delta_{i,n}$ yield that the algorithm's lower density under the $k$ orders is at least
\begin{align*}
\min_{i\in[k]}\liminf_{n\to\infty}\frac{\mu_{i,n}(D\cap K)}{\mu_{i,n}(K)}
&\geq
\min_{i\in[k]}\liminf_{n\to\infty}
\frac12 h(m) \cdot \frac{n-M_0-(km+2)\cdot \left(\log_2\frac{n}{m}+1\right)}{n}\\
&=\frac12 h(m).\qedhere
\end{align*}
\end{proof}

\subsection{Combining Randomization and Multiple Orders}
So far, we have analyzed randomized algorithms against non-adaptive adversaries and deterministic algorithms against adaptive adversaries under multiple orders. The following result combines the settings and techniques from these two sections.

\begin{definition}[Lower Density of a Randomized Algorithm under $k$ Orders]
The lower density of a randomized algorithm under $k$ orders $\sigma_1, \ldots, \sigma_k$ is
    \[
    \min_{i \in [k]}\liminf_{n \to \infty}\frac{\E{}{\mu_{i, n}(D \cap K)}}{\mu_{i, n}(K)}.
    \]
\end{definition}

\begin{theorem}[Lower Density of a Randomized Algorithm under $k$ Orders]\label{thm:rand_grouped_ps_algo_lower_density}
There is a randomized algorithm that, against every non-adaptive adversary, achieves lower density $1 - 1/e$ under $k$ orders.
\end{theorem}

The proof introduces no new ideas; it combines the \nameref{box:rand_ps_algo} with the \nameref{box:paired_ps_algo}. We defer the details to \cref{apx:rand_grouped_ps_algo_lower_density}.

\paragraph{Acknowledgments.}
We thank Moses Charikar, Chirag Pabbaraju, and Fan Wei for their helpful insights.
P.Z.\ was partially supported by NSF Grant CCF-2238682.

\paragraph{AI methodology and disclosure.}
The problem formulation and the main results are due to the authors, who also did essentially all of the proof work.
We proved the simple $1/2$ density result (\cref{thm:ps_algo_lower_density}) in January 2026, and from January to July 2026 we held weekly meetings to develop the other results and write the paper.
Our own proof drafts were complete enough to check for correctness, and we used GPT-5.5, GPT-6 Astra, and Claude Opus-5.5 to draft routine low-level steps in several parts of the proofs.
Beyond this, the main way AI helped with the proofs is that GPT-5.5 pointed out connections between our results and online bipartite matching and low-crossing partitions, which the paper uses.
In our view, these connections shorten the corresponding proofs and provide useful context, although the results also have more direct (if less elegant) proofs that do not use them.

We also used these models to draft and revise parts of the exposition, to generate TikZ code for figures, and to proofread and check earlier drafts.

The authors reviewed, edited, and checked all AI-assisted content and take full responsibility for the paper.

\bibliographystyle{alpha}
\bibliography{ref}

\appendix
\crefalias{section}{appendix}

\section{Low-Crossing Partition}\label{apx:low_crossing}

\begin{proof}[Proof of \cref{thm:low_crossing}]
Let $r=m/s$. We construct a binary partition tree. The root, at height $0$, is the set $[m]$. Consider a vertex at height $h$ associated with a set $A$ of size $ps$. If $p=1$, then this vertex is a leaf. Otherwise, let $j\in[k]$ be the index with $j\equiv h+1 \pmod{k}$. We split $A$ into two sets $A^-$ and $A^+$, where $A^-$ consists of the first $\lfloor p/2\rfloor s$ elements of $A$ in the order $\sigma_j$, and $A^+=A\setminus A^-$. Thus $|A^+|=\lceil p/2\rceil s$. The leaves of the tree form a balanced partition of $[m]$ into $r$ groups of size $s$.

It remains to bound the crossing number. Fix $i\in[k]$ and $n\in[m]$, and consider the cut $([n]_i,[m]\setminus[n]_i)$. Let $d=\lceil\log_2 r\rceil$. Every root-to-leaf path has length at most $d$. At a level where the tree splits according to $\sigma_i$, a crossed vertex gives rise to at most one crossed child, because the two children are consecutive in the $\sigma_i$-order restricted to that vertex. At any other level, a crossed vertex gives rise to at most two crossed children.

Among any $d$ consecutive levels, at least $\lfloor d/k\rfloor$ levels split according to $\sigma_i$. Therefore,
\[
c_i(n) \leq 2^{d-\lfloor d/k\rfloor} \leq 2^{d(1-1/k)+1} \leq 4r^{1-1/k}
\]
where the last inequality uses $d\leq \log_2 r+1$. This proves the theorem.
\end{proof}

\section{Language Generation from Partial Enumeration}\label{apx:ps_algo_partial}
In this section, we give the algorithm for the partial enumeration setting described at the end of \cref{sec:ps_algo} and prove its lower-density guarantee.

Directly applying the \nameref{box:ps_algo} cannot achieve any positive lower density in the worst case. This is because the validity of the algorithm relies on the fact that the adversary will enumerate $K$, which need not hold in the partial enumeration setting; \cref{ex:ps_algo_partial} gives a counterexample.

\begin{example}\label{ex:ps_algo_partial} Let $L_1$ be $\set{2k - 1 \mid k \in \Z^+} \cup \set{4k \mid k \in \Z^+}$, and let $L_i$ be $\set{2k \mid k \geq i - 1}$ for every integer $i > 1$. The adversary chooses $L_2$, the set of all positive even integers, as the true language $K$, but enumerates only the multiples of $4$ in increasing order.

Since $L_1$ remains consistent and $L_i \not\subseteq L_1$ for every integer $i > 1$, $L_1$ is the only critical language at every time step. Consequently, the \nameref{box:ps_algo} outputs only positive odd integers, which have zero density in the true language $K$. In fact, the algorithm does not even generate in the limit the true language $K$.
\end{example}

We can fix this issue without significantly changing the \nameref{box:ps_algo}. Given $\mathcal{L} = \set{L_1, L_2, L_3, \dots}$, construct a new collection $\mathcal{L}^\prime = \set{L_{\set{1}}^\prime, L_{\set{2}}^\prime, L_{\set{1, 2}}^\prime, \dots}$. Each index $I$ for the new language is a subset of indices for the old language. We define $L_I^\prime$ using finite intersection:
\[
L_I^\prime \defeq \bigcap_{i \in I}L_i.
\] We write the set $I$ as a $1$-based indicator binary. For example, the language $L_5^\prime$, whose index is ${101}_{(2)}$ in binary, is the intersection of $L_1$ and $L_3$, namely $L_{\set{1, 3}}^\prime$.

Each constructed language is countable, since it is a finite intersection of countable sets. If some $L_i^\prime$ is finite, we can simply discard it; to simplify the exposition, assume that every $L_i^\prime$ is countably infinite. The constructed collection is countable because it is indexed by finite subsets of a countable set.

The algorithm for the partial enumeration setting works as follows: given $\mathcal{L}$, it constructs the new collection $\mathcal{L}^\prime$ as stated above, and then runs the \nameref{box:ps_algo} on this new collection $\mathcal{L}^\prime$.

We first show that the adapted algorithm generates the true language $K$ in the limit.

\begin{lemma}[Validity]\label{lem:ps_algo_partial_valid} The \nameref{box:ps_algo} over $\mathcal{L}^\prime$ generates the true language $K$ in the limit.
\end{lemma}

\begin{proof}[Proof of \cref{lem:ps_algo_partial_valid}] Let $L_{i^\star}$ be the true language $K$. It suffices to show that, after sufficiently many time steps, there is an index $I$ with $2^{i^\star-1}\leq I < 2^{i^\star}$ such that $L_I^\prime$ is always critical. Such indices are precisely those whose binary representation has $i^\star$ as its highest nonzero bit, and hence $L_I^\prime \subseteq L_{i^\star}=K$.

Consider the languages $L_1^\prime, L_2^\prime, \dots, L_{2^{i^\star-1}}^\prime$. Since $L_{2^{i^\star-1}}^\prime = L_{i^\star}$ is always consistent, this finite block contains at least one critical language. Let $L_J^\prime$ be a critical language with the highest index $J$ among them. After finitely many steps, $J$ will stop changing.
Indeed, among $L_1^\prime, L_2^\prime, \dots, L_{2^{i^\star-1}}^\prime$, any language that is not a superset of $K^\prime$ will become inconsistent after sufficiently many time steps, while the consistency of all other languages will remain unchanged thereafter.

Once $J$ no longer changes, if $J = 2^{i^\star-1}$, the lemma follows immediately. Otherwise, consider the languages
$L_{J + 1}^\prime, L_{J + 2}^\prime, \dots, L_{J + 2^{i^\star-1}}^\prime$.
At least one of them must be critical: the language
$L_{J + 2^{i^\star-1}}^\prime = L_J^\prime \cap L_{i^\star}$ is consistent and contained in $L_J^\prime$. Applying the same argument, we obtain a language in this range that remains critical after sufficiently many time steps. Its index lies between $2^{i^\star-1}$ and $2^{i^\star} - 1$.

Consequently, there is a time step $T^\star$ such that every focus after time step $T^\star$ is a subset of $K$.
\end{proof}

\begin{theorem}[Lower Density from Partial Enumeration]\label{thm:ps_algo_partial_lower_density} If the adversary's enumerated set $K^\prime$ has lower density $\alpha$ in the true language $K$, then applying the \nameref{box:ps_algo} to $\mathcal{L}^\prime$ achieves lower density $\alpha / 2$.
\end{theorem}

As in the full-enumeration setting, we map $K$ to $\Z^+$ (with the ordering in $K$ preserved) and assume $K = \Z^+$.

\begin{proof}[Proof of \cref{thm:ps_algo_partial_lower_density}] Let $T^\star$ be as in the proof of \cref{lem:ps_algo_partial_valid}, and let $M_0$ be the number of integers first announced by the adversary up to time step $T^\star + 1$. As in the proof of \cref{thm:ps_algo_lower_density}, \cref{fac:ps_one_to_one,lem:ps_algo_switch_loss,lem:ps_algo_partial_valid} imply that, for every integer $n$,
  \[
  \mu_n(D \cap K) \geq \mu_n(A) - M_0 - \log_2n.
  \]
  It follows that
  \begin{align*}
\liminf_{n \to \infty}\frac{\mu_n(D\cap K)}{\mu_n(K)}
&=\liminf_{n \to \infty}\frac{2\mu_n(D\cap K)}{2n} \\
&\geq \liminf_{n \to \infty}\frac{\mu_n(D \cap K) + \mu_n(A) - M_0 - \log_2n}{2n} \\
&\geq\liminf_{n \to \infty}\frac{\mu_n(K^\prime) - M_0 - \log_2n}{2n} &\tag{$K^\prime \subseteq (D \cap K) \sqcup A$}\\
&= \frac12\liminf_{n \to \infty}\left(\frac{\mu_n(K^\prime)}{n} - \frac{M_0 + \log_2n}{n}\right)\\
&= \frac\alpha2.\qedhere
  \end{align*}
\end{proof}

\section{A Tight Randomized Algorithm under a Single Order}\label{apx:rand_ps_algo_lower_density}
In this section, we show how to adapt the \nameref{box:rand_ps_algo} to achieve the tight lower density $1 - 1/e$.

The \nameref{box:rand_ps_algo} uses a fixed block size $m$, leaving a gap between $\left(1 - \frac1{m + 1}\right) \cdot f(m + 1)$ and $1 - 1/e$ for any fixed $m$. To close this gap, we allow the block size to grow dynamically.

When a new block is reserved, let $a$ be the smallest available element of the current focus $F$. The block then has $\lceil\sqrt m\rceil$ elements, where $m$ is the rank of $a$ within $F$. This ensures that most blocks intersecting $[n]$ are large, while preventing any incomplete block within $[n]$ from becoming too large and affecting the algorithm's density.

\begin{mybox}[label={box:rand_ps_algo_var},nameref={variable-block randomized patient-scope algorithm}]
{The variable-block randomized patient-scope algorithm}
\begin{itemize}
    \item Let $N_B=0$ be the current number of blocks.
    \item Initially, at time $t=0$, set the scope size $s_0=1$, set the focus-change
    count $\tau=1$, and create no blocks.
    \item In each time step $t$, perform the following operations.
    \begin{itemize}
        \item Set the scope size $s_t = s_{t - 1}$.
        \item Receive a new integer from the adversary, and update the sets of consistent and critical languages accordingly.
        \item If the focus becomes inconsistent, run the \nameref{box:bt_algo}.
        \item Otherwise, if the focus is still consistent and the algorithm has reserved $2^\tau$ blocks since the last focus change, increase the scope size by one and update
        the focus accordingly. If the focus changes, increase $\tau$ by one.
        \item Let $x_t$ be the adversary's announced number in this step.
        \begin{itemize}
            \item If $x_t\in B_r$ for some block $B_r$ and $B_r$ contains a number not yet announced by either party,
            output the first such number in the permutation $\pi_r$.
            \item Otherwise, let $a$ be the smallest number in the current focus $F$ that has not yet been announced by either party or reserved by the algorithm, and let $m\defeq\abs{F\cap[a]}$ be the rank of $a$ within $F$. Let \[B_{N_B+1}\defeq (W \cap F)[\lceil\sqrt m\rceil].\] Namely, block $B_{N_B+1}$ consists of the $\lceil\sqrt{m}\rceil$ smallest available elements of the current focus $F$. Generate an independent uniformly
            random permutation $\pi_{N_B+1}$ of $B_{N_B+1}$, mark all elements of
            $B_{N_B+1}$ as reserved, and output the first number of $\pi_{N_B+1}$, and increase $N_B$ by $1$.
        \end{itemize}
    \end{itemize}
\end{itemize}
\end{mybox}

The following lemma is analogous to \cref{lem:rand_ps_algo_validity}, so we state it without proof.

\begin{lemma}[Validity]\label{lem:rand_ps_algo_var_validity} Regardless of the realization of random permutations, the \nameref{box:rand_ps_algo_var} can always generate in the limit the true language $K$.
\end{lemma}

As in the proof of \cref{lem:rand_ps_algo_validity}, there is a time step $T^\star$ such that every focus after time step $T^\star$ is a subset of $K$. This time step may depend on the realization of the random permutations.

As in \cref{sec:rand_ps_algo}, we fix an oblivious adversary, and all probabilities and expectations are taken over the random permutations of the algorithm. The algorithm creates infinitely many blocks $B_1, B_2, \ldots$, indexed in the order of creation; let $t_r$ be the time step at which $B_r$ is created, assume that $\pi_r$ is drawn from an independent random seed $\xi_r$, and let $\mathcal{H}_r$ denote the $\sigma$-algebra generated by $\xi_1, \ldots, \xi_{r-1}$. As before, everything that happens before $\pi_r$ is first used is determined by $\xi_1, \ldots, \xi_{r-1}$.

For a block $B_r$, let $F_r$ denote the focus at time step $t_r$, which we call the \emph{focus of $B_r$}, let $a_r$ be the smallest element of $W \cap F_r$ at time step $t_r$, let $m_r = \abs{F_r \cap [a_r]}$ be its rank within $F_r$, and let $b_r = \abs{B_r} = \lceil\sqrt{m_r}\rceil$.

As in the proof of \cref{thm:rand_ps_algo_lower_density_weaker}, we identify $K$ with $\Z^+$ in an order-preserving way, so that $[n]$ denotes the set of the $n$ smallest integers of $K$. Integers outside $K$ can be reserved only in the finitely many blocks created up to time step $T^\star$. For every block $B_r$ created after time step $T^\star$, we have $F_r \subseteq K$, so the identification preserves the rank of each element within $F_r$, and hence $m_r$ and $b_r$.

We define triggers and the sets $C_r$ as in \cref{sec:rand_ps_algo}. The sets $C_r$ are pairwise disjoint and cover $K$, so \eqref{eq:rand_ps_algo_partition} holds. Whether $B_r$ is complete within $n$ is determined at the creation of $B_r$, so the event $\set{B_r \subseteq [n]}$ is $\mathcal{H}_r$-measurable. We then classify each integer in $[n]$ into two cases:
\begin{enumerate}
    \item it belongs to $C_r$ for a block $B_r$ complete within $n$;
    \item it belongs to $C_r$ for a block $B_r$ incomplete within $n$.
\end{enumerate}

We first count the number of blocks involved in the second case.

\begin{lemma}\label{lem:rand_ps_algo_var_incmp_blks} For every realization and any positive integer $n$, there are at most $\log_2{n} + 1$ blocks $B_r$ created after time step $T^\star$ that are incomplete within $n$ and satisfy either $B_r\cap[n]\neq\emptyset$ or $B_r$ has a trigger at most $n$.
\end{lemma}

\begin{proof}[Proof of \cref{lem:rand_ps_algo_var_incmp_blks}]
As in the proof of \cref{lem:rand_ps_algo_incmp_blks}, after the first incomplete block, each later counted block corresponds to a unique switch loss $\ell$: either $\ell$ is an unreserved integer at most $n$ that triggers the new block, or the switch makes the smallest integer in $W \cap F$ at most $n$. In both cases, we charge $\ell$ to $2^{\tau^\prime}$ blocks reserved by the \nameref{box:rand_ps_algo_var} that lie below $n$. If there are $w$ such switch losses $\ell_1, \ell_2, \ldots, \ell_w$, with focus-change counts $\tau_1, \tau_2, \ldots, \tau_w$, then disjointness of the charged sets implies
    \[
    n \geq 2^{\tau_1} + 2^{\tau_2} + \cdots + 2^{\tau_w}.
    \]
    Since $\tau_1, \tau_2, \ldots, \tau_w$ are distinct positive integers, we have $n \geq 2^1+\cdots+2^w \geq 2^w$. Taking logarithms gives $w \leq \log_2 n$. The first incomplete block accounts for the additional $1$.
\end{proof}

\begin{lemma}\label{lem:rand_ps_algo_var_few_foci}
For every realization and any positive integer $n$, the blocks $B_r$ that are created after time step $T^\star$ and satisfy $B_r\cap[n]\neq\emptyset$ have at most $2\log_2 n+1$ distinct foci.
\end{lemma}

\begin{proof}[Proof of \cref{lem:rand_ps_algo_var_few_foci}]
A block is \emph{relevant} if it is created after time step $T^\star$ and intersects $[n]$; a \emph{stretch} is a maximal time interval without switch losses. Within a stretch, each new focus is a subset of the previous one, and $W$ only shrinks.

\emph{Claim:} if a block $B$ is reserved before a relevant block $B'$ in the same stretch, then $B\subseteq[n]$. Indeed, the smallest element $a'\leq n$ of $B'$ was in $W$ and in $B$'s focus when $B$ was reserved, yet $B$ skipped it; so $B$ lies below $a'$.

We count the maximal constant-focus periods containing a relevant block, which are of two types.

\emph{Followed by another such period in the same stretch.} This period ended with a scope increase, after the algorithm reserved $2^{\tau'}$ blocks since the previous focus change; these blocks precede the next relevant block, so lie in $[n]$ by the claim. Since the counts $\tau'$ are distinct and the blocks disjoint, $w$ such periods give $n\geq 2^w$, so $w\leq\log_2 n$.

\emph{Last in its stretch.} There is one such period per stretch containing a relevant block, so we count these stretches. Each except the first begins with a switch loss; let $F$ be the focus right after it. The smallest element $a'\leq n$ of a relevant block in the stretch was already in $W\cap F$ then, since later foci are subsets of $F$ and $W$ only shrinks. Hence $\min(W\cap F)\leq n$ right after the switch loss. The proof of \cref{lem:rand_ps_algo_var_incmp_blks} bounds such switch losses by $\log_2 n$, so there are at most $\log_2 n+1$ such stretches.

In total, there are at most $2\log_2 n+1$ distinct foci.
\end{proof}

Next, we compute the algorithm's density among integers in the first case.

For $b \geq 1$, define
\[
\phi(b)=\left(1-\frac1{b+1}\right)\cdot f(b+1)
=\left(1-\frac1{b+1}\right)\cdot \left(1-\left(1-\frac1{b+2}\right)^{b+1}\right).
\]

\begin{lemma}\label{lem:rand_ps_algo_var_cmp_blks}
For any positive integer $n$ and any block $B_r$, on the event that $B_r$ is complete within $n$,
\[
\E{}{\abs{D \cap B_r} \mid \mathcal{H}_r} \geq \phi(b_r) \cdot \abs{C_r \cap [n]}.
\]
\end{lemma}

\begin{proof}[Proof of \cref{lem:rand_ps_algo_var_cmp_blks}]
The proof of \cref{lem:rand_ps_algo_cmp_blks} applies verbatim, with $m$ replaced by $b_r$, which is determined by $\mathcal{H}_r$. Namely, conditional on $\mathcal{H}_r$, the algorithm's announcements within $B_r$ correspond exactly to the matches produced by Ranking with permutation $\pi_r$ on a graph that is determined by $\mathcal{H}_r$, has $b_r + 1$ online vertices, and has a matching of size $b_r$. Thus, \cref{thm:ranking_algo} yields $\E{}{\abs{D \cap B_r} \mid \mathcal{H}_r} \geq b_r \cdot f(b_r + 1)$, and the claim follows from $\abs{C_r \cap [n]} \leq \abs{C_r} \leq b_r + 1$.
\end{proof}

We have $\phi(b)\to 1-1/e$ as $b\to\infty$. Hence, for every $\varepsilon > 0$, there exists an integer $T=T(\varepsilon)$ such that $\phi(b)\geq 1-1/e-\varepsilon$ for every $b\geq T$. It remains to control the contribution of blocks with $b_r < T$.

\begin{lemma}\label{lem:rand_ps_algo_var_small_blks}
Fix any $\varepsilon > 0$, and let
\[
C_\varepsilon\defeq\sum_{b=1}^{T-1}(b+1)(2b-1).
\]
For every realization and any positive integer $n$, the blocks $B_r$ that are created after time step $T^\star$ and satisfy $B_r \cap [n] \neq \emptyset$ and $b_r < T$ satisfy
\[
\sum_{r} \abs{C_r \cap [n]} \leq C_\varepsilon\cdot(2\log_2 n+1),
\]
where the sum ranges over these blocks.
\end{lemma}

\begin{proof}[Proof of \cref{lem:rand_ps_algo_var_small_blks}]
For a fixed $b\geq 1$, the condition $\lceil\sqrt{m_r}\rceil=b$ is equivalent to
\[
(b-1)^2 < m_r \leq b^2.
\]
Fix a language $F$. Every block $B_r$ with focus $F$ contains $a_r$, so $a_r$ is reserved when $B_r$ is created. Hence different blocks with focus $F$ have different elements $a_r\in F$, and therefore different ranks $m_r$ within $F$. Consequently, for every $b\geq 1$, at most $2b-1$ blocks with focus $F$ satisfy $b_r=b$. Since $\abs{C_r} \leq b_r + 1$, the blocks with focus $F$ and $b_r<T$ contribute at most $C_\varepsilon$ to the sum. By \cref{lem:rand_ps_algo_var_few_foci}, the blocks in the sum have at most $2\log_2 n+1$ distinct foci. The result follows.
\end{proof}

Finally, we can compute the lower density of the \nameref{box:rand_ps_algo_var}.

\begin{proof}[Proof of \cref{thm:rand_ps_algo_lower_density}]
Fix any $\varepsilon \in (0, 1 - 1/e)$, let $T = T(\varepsilon)$ and $C_\varepsilon$ be as above, and fix an arbitrary prefix $[n]$. For each block $B_r$, let $I_r \defeq \mathbf{1}\set{B_r \subseteq [n]}$ indicate that $B_r$ is complete within $n$, and let $J_r \defeq \mathbf{1}\set{b_r \geq T}$; both are $\mathcal{H}_r$-measurable. Let
\[
\Delta_n \defeq \sum_r (1 - I_r) \cdot \abs{C_r \cap [n]}
\qquad\text{and}\qquad
S_n \defeq \sum_r I_r \cdot (1 - J_r) \cdot \abs{C_r \cap [n]}
\]
be the number of integers in the second case and the number of integers in the first case associated with blocks with $b_r < T$, respectively. By \eqref{eq:rand_ps_algo_partition}, $\sum_r I_r \cdot J_r \cdot \abs{C_r \cap [n]} = n - \Delta_n - S_n$.

\begin{enumerate}
    \item For the integers in the first case: since the blocks are pairwise disjoint, we have $\mu_n(D \cap K) \geq \sum_r I_r \cdot J_r \cdot \abs{D \cap B_r}$. Since $I_r$, $J_r$, $b_r$, and $\abs{C_r \cap [n]}$ are $\mathcal{H}_r$-measurable, the law of total expectation and \cref{lem:rand_ps_algo_var_cmp_blks} give, for every $r$,
    \[
    \E{}{I_r \cdot J_r \cdot \abs{D \cap B_r}} = \E{}{I_r \cdot J_r \cdot \E{}{\abs{D \cap B_r} \mid \mathcal{H}_r}} \geq \left(1 - \frac1e - \varepsilon\right) \cdot \E{}{I_r \cdot J_r \cdot \abs{C_r \cap [n]}},
    \]
    where we used $\phi(b_r) \geq 1 - 1/e - \varepsilon$ whenever $J_r = 1$. Summing over $r$, and using monotone convergence for these series of nonnegative terms, we obtain
    \begin{equation}\label{eq:rand_ps_algo_var_first_case}
    \E{}{\mu_n(D \cap K)} \geq \left(1 - \frac1e - \varepsilon\right) \cdot \E{}{n - \Delta_n - S_n}.
    \end{equation}
    \item For the remaining integers, fix a realization, and let $M_0$ be the total number of integers in the sets $C_r$ of the blocks created up to time step $T^\star$. By \cref{lem:rand_ps_algo_var_validity}, $M_0$ is finite and does not depend on $n$.

    Every block created after time step $T^\star$ contributes to $\Delta_n$ only if it is incomplete within $n$ and $C_r \cap [n] \neq \emptyset$. By \cref{lem:rand_ps_algo_var_incmp_blks}, there are at most $\log_2 n+1$ such blocks. If such a block $B_r$ intersects $[n]$, pick $y\in B_r\cap[n]$. The number $y$ was available when $B_r$ was created, so $a_r\leq y$; moreover, $F_r\subseteq K=\Z^+$. Hence $m_r=\abs{F_r\cap[a_r]}\leq\abs{F_r\cap[y]}\leq y\leq n$, and thus the block has size at most $\lceil\sqrt n\rceil$. Otherwise, the trigger is the only integer of $C_r$ that lies in $[n]$. Either way, such a block contributes at most $\lceil\sqrt n\rceil+1$ integers.

    Every block created after time step $T^\star$ contributes to $S_n$ only if it is complete within $n$, and hence intersects $[n]$, and satisfies $b_r < T$. By \cref{lem:rand_ps_algo_var_small_blks}, these blocks contribute at most $C_\varepsilon \cdot (2\log_2 n + 1)$ integers. Therefore,
    \[
    \Delta_n + S_n \leq M_0 + (\lceil\sqrt n\rceil+1)\cdot(\log_2 n+1) + C_\varepsilon\cdot(2\log_2 n+1),
    \]
    and hence $(n - \Delta_n - S_n)/n \to 1$ as $n \to \infty$, for every realization.
\end{enumerate}

Note that $T^\star$, and hence $M_0$, may depend on the algorithm's random choices, so we pass to the limit using Fatou's lemma. Since $\mu_n(K) = n$ and $(n - \Delta_n - S_n)/n \geq 0$, \eqref{eq:rand_ps_algo_var_first_case} and Fatou's lemma yield that the algorithm's lower density is at least
\begin{align*}
\liminf_{n\to\infty}\frac{\E{}{\mu_n(D\cap K)}}{\mu_n(K)}
&\geq
\left(1-\frac1e-\varepsilon\right) \cdot \liminf_{n \to \infty}\E{}{\frac{n - \Delta_n - S_n}{n}}\\
&\geq \left(1-\frac1e-\varepsilon\right) \cdot \E{}{\liminf_{n \to \infty}\frac{n - \Delta_n - S_n}{n}}
=1-\frac1e-\varepsilon.
\end{align*}
Since $\varepsilon \in (0, 1 - 1/e)$ was arbitrary, the lower density is at least $1-1/e$.
\end{proof}

\section{A Tight Deterministic Algorithm under Multiple Orders}\label{apx:paired_ps_algo_lower_density}
In this section, we adapt the \nameref{box:paired_ps_algo} to achieve the tight lower density $1/2$. The idea is the same as in the \nameref{box:rand_ps_algo_var}: choose the block size dynamically, according to the smallest available rank in the current focus.

For a set $F$ of integers, a nonempty set $S\subseteq F$, and an order index $i\in[k]$, let
\[
\rho_i^F(S) \defeq \min_{x\in S}\abs{\set{y\in F : \sigma_i^{-1}(y)\leq\sigma_i^{-1}(x)}},
\]
so that $\rho_i^F(S)$ is the smallest rank, within $F$ under order $\sigma_i$, of an element of $S$.

\begin{mybox}[label={box:paired_ps_algo_var},nameref={variable-block paired patient-scope algorithm}]
{The variable-block paired patient-scope algorithm}
\begin{itemize}
    \item Let $N_B=0$ be the current number of blocks.
    \item Initially, at time $t=0$, set the scope size $s_0=1$, set the focus-change
    count $\tau=1$, and create no blocks.
    \item In each time step $t$, perform the following operations.
    \begin{itemize}
        \item Set the scope size $s_t = s_{t - 1}$.
        \item Receive a new integer from the adversary, and update the sets of consistent and critical languages accordingly.
        \item If the focus becomes inconsistent, run the \nameref{box:bt_algo}.
	        \item Otherwise, if the focus $F$ is still consistent and the algorithm has reserved $2^\tau$ blocks since the last focus change, increase the scope size by one and update
	        the focus accordingly. If the focus changes, increase $\tau$ by one.
	        \item Let $x_t$ be the adversary's announced number in this step.
	        \begin{itemize}
	            \item If $x_t\in B_r$ for some block $B_r$ and $P_r(x_t)$ has not yet been announced
	            by either party, output $P_r(x_t)$.
	            \item Otherwise, set
	            \[
	            m \defeq \min_{i\in[k]}\rho_i^F(W\cap F).
	            \]
	            Define
	            \[
	            B_{N_B+1} \defeq \bigcup_{i \in [k]} (W \cap F)[\lceil\sqrt{m}\rceil]_i.
	            \] Namely, block $B_{N_B+1}$ contains, for each order, the $\lceil\sqrt{m}\rceil$ smallest available elements of the current focus $F$ under that order. If $\abs{B_{N_B + 1}}$ is odd, we further add the smallest element in $(W \cap F)\setminus B_{N_B + 1}$ under $\sigma_1$ to $B_{N_B + 1}$. Partition $B_{N_B+1}$ into pairs $P_{N_B+1}$ by applying
	            \cref{thm:low_crossing} to the induced orders of
	            $\sigma_1, \dots, \sigma_k$ on $B_{N_B+1}$. Mark all elements of $B_{N_B+1}$ as
	            reserved, and output the first element of $B_{N_B+1}$ under $\sigma_1$. Increase $N_B$
	            by $1$.
        \end{itemize}
    \end{itemize}
\end{itemize}
\end{mybox}

The following lemma is analogous to \cref{lem:paired_ps_algo_validity}, so we state it without proof.

\begin{lemma}[Validity]\label{lem:paired_ps_algo_var_validity} The \nameref{box:paired_ps_algo_var} can generate in the limit the true language $K$.
\end{lemma}

As in the proof of \cref{lem:paired_ps_algo_validity}, there is a time step $T^\star$ such that every focus after time step $T^\star$ is a subset of $K$.

For a block $B_r$, let $F_r$ be the focus and $m_r$ the value of $m$ when $B_r$ is created, and let $b_r=\lceil\sqrt{m_r}\rceil$.

As in the proof of \cref{thm:paired_ps_algo_lower_density_weaker}, we identify $K$ with $\Z^+$ in an order-preserving way and update each $\sigma_i$ accordingly, so that $[n]_i$ denotes the set of the first $n$ integers of $K$ under $\sigma_i$. Integers outside $K$ can be reserved only in the finitely many blocks created up to time step $T^\star$. For every block $B_r$ created after time step $T^\star$, we have $F_r \subseteq K$, so the identification preserves the rank of each element within $F_r$ under each $\sigma_i$, and hence $m_r$ and $b_r$. We define triggers and the sets $C_r$ as in \cref{sec:paired_ps_algo}. The sets $C_r$ are pairwise disjoint and cover $K$, so \eqref{eq:paired_ps_algo_partition} holds.

Define the \emph{head under $\sigma_i$} of $B_r$ to be the set $B_r[b_r]_i$. During any consecutive interval with no switch loss, the heads of newly reserved blocks are strictly above the heads of older reserved blocks under $\sigma_i$.

Fix a prefix $[n]_i$. A head $B_r[b_r]_i$ is \emph{complete within $n$} if $B_r[b_r]_i\subseteq[n]_i$; otherwise it is \emph{incomplete within $n$}. Each integer in $[n]_i$ falls into one of the following two cases:
\begin{enumerate}
    \item it belongs to $C_r$ for a block $B_r\subseteq K$ whose head $B_r[b_r]_i$ is complete within $n$;
    \item it belongs to $C_r$ for any other block $B_r$.
\end{enumerate}

We first count the number of blocks involved in the second case.

\begin{lemma}\label{lem:paired_ps_algo_var_incmp_heads} For any prefix $[n]_i$, there are at most $\log_2 n + 1$ blocks $B_r$ created after time step $T^\star$ whose heads $B_r[b_r]_i$ are incomplete within $n$ and that satisfy either $B_r\cap[n]_i\neq\emptyset$ or $B_r$ has a trigger in $[n]_i$.
\end{lemma}

\begin{proof}[Proof of \cref{lem:paired_ps_algo_var_incmp_heads}]
As in the proof of \cref{lem:paired_ps_algo_incmp_heads}, after the first block with an incomplete head, each later counted block corresponds to a unique switch loss $\ell$: either $\ell$ is an unreserved element of $[n]_i$ that triggers the new block, or the switch makes the first element of $W\cap F$ under $\sigma_i$ lie in $[n]_i$. In both cases, we charge $\ell$ to $2^{\tau^\prime}$ heads reserved by the \nameref{box:paired_ps_algo_var} that lie in $[n]_i$.

If there are $w$ such switch losses $\ell_1,\ell_2,\ldots,\ell_w$ with corresponding focus-change counts $\tau_1,\tau_2,\ldots,\tau_w$, then disjointness of the charged sets implies
\[
n \geq 2^{\tau_1}+2^{\tau_2}+\cdots+2^{\tau_w}.
\]
Since $\tau_1, \tau_2, \ldots, \tau_w$ are distinct positive integers, we have $n\geq 2^1+\cdots+2^w\geq 2^w$. Taking logarithms gives $w\leq \log_2 n$. The first block with an incomplete head accounts for the additional $1$.
\end{proof}

\begin{lemma}\label{lem:paired_ps_algo_var_few_foci}
For any prefix $[n]_i$, the blocks $B_r$ that are created after time step $T^\star$ and satisfy $B_r\cap[n]_i\neq\emptyset$ have at most $2\log_2 n+1$ distinct foci $F_r$.
\end{lemma}

\begin{proof}[Proof of \cref{lem:paired_ps_algo_var_few_foci}]
A block is \emph{relevant} if it is created after time step $T^\star$ and intersects $[n]_i$; stretches are defined as in the proof of \cref{lem:rand_ps_algo_var_few_foci}. Within a stretch, each new focus is a subset of the previous one, and $W$ only shrinks.

\emph{Claim:} if a block $B_r$ is reserved before a relevant block $B'$ in the same stretch, then its head $B_r[b_r]_i$ lies in $[n]_i$. Indeed, since $(W\cap F_r)[b_r]_i\subseteq B_r\subseteq W\cap F_r$, the head is the set of the first $b_r$ elements of $W\cap F_r$ under $\sigma_i$. An element $y\in B'\cap[n]_i$ was in $W\cap F_r$ when $B_r$ was reserved, yet $B_r$ skipped it; so the head lies before $y$ under $\sigma_i$.

We count the maximal constant-focus periods containing a relevant block, which are of two types.

\emph{Followed by another such period in the same stretch.} This period ended with a scope increase, after the algorithm reserved $2^{\tau'}$ blocks since the previous focus change; these blocks precede the next relevant block, so their heads lie in $[n]_i$ by the claim. Since the counts $\tau'$ are distinct and the heads nonempty and disjoint, $w$ such periods give $n\geq 2^w$, so $w\leq\log_2 n$.

\emph{Last in its stretch.} There is one such period per stretch containing a relevant block, so we count these stretches. Each except the first begins with a switch loss; let $F$ be the focus right after it. An element $y\in[n]_i$ of a relevant block in the stretch was already in $W\cap F$ then, since later foci are subsets of $F$ and $W$ only shrinks. Hence the first element of $W\cap F$ under $\sigma_i$ lies in $[n]_i$ right after the switch loss. The proof of \cref{lem:paired_ps_algo_var_incmp_heads} bounds such switch losses by $\log_2 n$, so there are at most $\log_2 n+1$ such stretches.

In total, there are at most $2\log_2 n+1$ distinct foci.
\end{proof}

Next, we compute the algorithm's density among integers in the first case.

For $b\geq 1$, define
\[
\phi(b)=\frac12\cdot\left(1-\frac{4\cdot(kb+1)^{1-1/k}+1}{b}\right).
\]

\begin{lemma}\label{lem:paired_ps_algo_var_cmp_heads}
For any order index $i$, any positive integer $n$, and any block $B_r\subseteq K$ whose head $B_r[b_r]_i$ is complete within $n$,
\[
\abs{D \cap B_r \cap [n]_i} \geq \phi(b_r) \cdot \abs{C_r \cap [n]_i}.
\]
\end{lemma}

\begin{proof}[Proof of \cref{lem:paired_ps_algo_var_cmp_heads}]
The proof of \cref{lem:paired_ps_algo_cmp_heads} applies verbatim, with $m$ replaced by $b_r$: since $B_r\subseteq K$, the set $B_r\cap[n]_i$ is a prefix of $B_r$ under $\sigma_i$, and the block $B_r$ has size at most $kb_r+1$, so by \cref{thm:low_crossing} at most $4\cdot(kb_r+1)^{1-1/k}$ of its pairs are crossed by the cut $(B_r \cap [n]_i, B_r \setminus [n]_i)$; and since the head is complete within $n$, we have $\abs{C_r \cap [n]_i} \geq b_r$.
\end{proof}

We have $\phi(b)\to 1/2$ as $b\to\infty$. Hence, for every $\varepsilon > 0$, there exists an integer $T=T(\varepsilon)$ such that $\phi(b)\geq 1/2-\varepsilon$ for every $b\geq T$. It remains to control the contribution of blocks with $b_r < T$.

\begin{lemma}\label{lem:paired_ps_algo_var_small_heads}
Fix any $\varepsilon > 0$, and let
\[
C_\varepsilon\defeq\sum_{b=1}^{T-1}(kb+2)\cdot k(2b-1).
\]
For any prefix $[n]_i$, the blocks $B_r$ that are created after time step $T^\star$ and satisfy $B_r \cap [n]_i \neq \emptyset$ and $b_r < T$ satisfy
\[
\sum_{r} \abs{C_r \cap [n]_i} \leq C_\varepsilon\cdot(2\log_2 n+1),
\]
where the sum ranges over these blocks.
\end{lemma}

\begin{proof}[Proof of \cref{lem:paired_ps_algo_var_small_heads}]
At the time $B_r$ is created, some order $j\in[k]$ satisfies $m_r=\rho_j^{F_r}(W\cap F_r)$, and the corresponding first available element of $F_r$ under $\sigma_j$, whose rank within $F_r$ under $\sigma_j$ equals $m_r$, is reserved in $B_r$. Hence, for each language $F$, each order $j\in[k]$, and each value of $m_r$, this can happen for at most one block with focus $F$. Thus, for a fixed language $F$ and a fixed $b\geq 1$, there are at most $k(2b-1)$ blocks with focus $F$ and $b_r=b$. Each such block satisfies $\abs{C_r} \leq kb+2$, including its trigger. Therefore, the blocks with focus $F$ and $b_r<T$ contribute at most $C_\varepsilon$ to the sum. By \cref{lem:paired_ps_algo_var_few_foci}, the blocks in the sum have at most $2\log_2 n+1$ distinct foci. The result follows.
\end{proof}

Finally, we can compute the lower density of the \nameref{box:paired_ps_algo_var}.

\begin{proof}[Proof of \cref{thm:paired_ps_algo_lower_density}]
Fix any $\varepsilon \in (0, 1/2)$, let $T = T(\varepsilon)$ and $C_\varepsilon$ be as above, and fix an order index $i$ and an arbitrary prefix $[n]_i$. For each block $B_r$, let $I_r \defeq \mathbf{1}\set{B_r \subseteq K \text{ and } B_r[b_r]_i \subseteq [n]_i}$ indicate that $B_r \subseteq K$ and the head $B_r[b_r]_i$ is complete within $n$, and let $J_r \defeq \mathbf{1}\set{b_r \geq T}$. Let
\[
\Delta_{i,n} \defeq \sum_r (1 - I_r) \cdot \abs{C_r \cap [n]_i}
\qquad\text{and}\qquad
S_{i,n} \defeq \sum_r I_r \cdot (1 - J_r) \cdot \abs{C_r \cap [n]_i}
\]
be the number of integers in the second case and the number of integers in the first case associated with blocks with $b_r < T$, respectively. By \eqref{eq:paired_ps_algo_partition}, $\sum_r I_r \cdot J_r \cdot \abs{C_r \cap [n]_i} = n - \Delta_{i,n} - S_{i,n}$.

\begin{enumerate}
    \item For the integers in the first case: since the blocks are pairwise disjoint, we have $\mu_{i,n}(D \cap K) \geq \sum_r I_r \cdot J_r \cdot \abs{D \cap B_r \cap [n]_i}$. By \cref{lem:paired_ps_algo_var_cmp_heads}, and since $\phi(b_r) \geq 1/2 - \varepsilon$ whenever $J_r = 1$, we obtain
    \begin{equation}\label{eq:paired_ps_algo_var_first_case}
    \mu_{i,n}(D \cap K) \geq \left(\frac12 - \varepsilon\right) \cdot \sum_r I_r \cdot J_r \cdot \abs{C_r \cap [n]_i} = \left(\frac12 - \varepsilon\right) \cdot \left(n - \Delta_{i,n} - S_{i,n}\right).
    \end{equation}
    \item For the remaining integers, let $M_0$ be the total number of integers in the sets $C_r$ of the blocks created up to time step $T^\star$. By \cref{lem:paired_ps_algo_var_validity}, $M_0$ is finite and depends on neither $i$ nor $n$.

    Every block created after time step $T^\star$ is contained in $K$, and hence contributes to $\Delta_{i,n}$ only if its head is incomplete within $n$ and $C_r \cap [n]_i \neq \emptyset$. By \cref{lem:paired_ps_algo_var_incmp_heads}, there are at most $\log_2 n+1$ such blocks. For each such block, the integers it contributes consist of at most one trigger, plus the elements of $B_r\cap[n]_i$. If $B_r\cap[n]_i\neq\emptyset$, pick $y\in B_r\cap[n]_i$. The element $y$ was available when $B_r$ was created, and $F_r\subseteq K=\Z^+$, so
    \[
    m_r\leq\rho_i^{F_r}(W\cap F_r)\leq\abs{\set{z\in F_r : \sigma_i^{-1}(z)\leq\sigma_i^{-1}(y)}}\leq\sigma_i^{-1}(y)\leq n.
    \]
    Hence $b_r\leq\lceil\sqrt n\rceil$ and $\abs{B_r}\leq k\lceil\sqrt n\rceil+1$. Otherwise, the block contributes only its trigger. Either way, such a block contributes at most $k\lceil\sqrt n\rceil+2$ integers.

    Every block created after time step $T^\star$ contributes to $S_{i,n}$ only if its head is complete within $n$, and hence it intersects $[n]_i$, and $b_r < T$. By \cref{lem:paired_ps_algo_var_small_heads}, these blocks contribute at most $C_\varepsilon \cdot (2\log_2 n + 1)$ integers. Therefore,
    \[
    \Delta_{i,n} + S_{i,n} \leq M_0 + (k\lceil\sqrt n\rceil+2)\cdot(\log_2 n+1) + C_\varepsilon\cdot(2\log_2 n+1),
    \]
    and hence $(n - \Delta_{i,n} - S_{i,n})/n \to 1$ as $n \to \infty$.
\end{enumerate}

Since $\mu_{i,n}(K) = n$, \eqref{eq:paired_ps_algo_var_first_case} yields that the algorithm's lower density under the $k$ orders is at least
\begin{align*}
\min_{i\in[k]}\liminf_{n\to\infty}\frac{\mu_{i,n}(D\cap K)}{\mu_{i,n}(K)}
&\geq
\min_{i\in[k]}\liminf_{n\to\infty}
\left(\frac12-\varepsilon\right)
\cdot
\frac{n-\Delta_{i,n}-S_{i,n}}{n}
=\frac12-\varepsilon.
\end{align*}
Since $\varepsilon \in (0, 1/2)$ was arbitrary, the lower density is at least $1/2$.
\end{proof}

\section{A Tight Randomized Algorithm under Multiple Orders}\label{apx:rand_grouped_ps_algo_lower_density}

The algorithm maintains a family of pairwise disjoint blocks $B_1, B_2, \dots$. Each block $B_r$ is partitioned into pairwise disjoint groups $G_{r, 1}, G_{r, 2}, \dots$. Each group $G_{r, r^\prime}$ is equipped with an independent uniformly random permutation $\pi_{r, r^\prime}$.

\begin{mybox}[label={box:rand_grouped_ps_algo},nameref={randomized grouped patient-scope algorithm}]
{The randomized grouped patient-scope algorithm}
\begin{itemize}
    \item Let $N_B=0$ be the current number of blocks.
    \item Initially, at time $t=0$, set the scope size $s_0=1$, set the focus-change
    count $\tau=1$, and create no blocks.
    \item In each time step $t$, perform the following operations.
    \begin{itemize}
        \item Set the scope size $s_t = s_{t - 1}$.
        \item Receive a new integer from the adversary, and update the sets of consistent and critical languages accordingly.
        \item If the focus becomes inconsistent, run the \nameref{box:bt_algo}.
        \item Otherwise, if the focus $F$ is still consistent and the algorithm has reserved $2^\tau$ blocks since the last focus change, increase the scope size by one and update
        the focus accordingly. If the focus changes, increase $\tau$ by one.
        \item Let $x_t$ be the adversary's announced number in this step.
        \begin{itemize}
            \item If $x_t\in G_{r, r^\prime}$ for some group $G_{r, r^\prime}$ and $G_{r, r^\prime}$ contains a number not yet announced by either party, output the first such number in the permutation $\pi_{r, r^\prime}$.
            \item Otherwise, let
            \[
            m \defeq \min_{i \in [k]}\rho_i^F(W \cap F).
            \]
            Define
            \[
            B_{N_B+1} \defeq \bigcup_{i \in [k]} (W \cap F)\left[\lceil m^{2/3}\rceil\right]_i.
            \] Namely, block $B_{N_B+1}$ contains, for each order, the $\lceil m^{2/3}\rceil$ smallest available elements of the current focus $F$ under that order. We further repeatedly add the smallest element in $(W \cap F)\setminus B_{N_B + 1}$ under $\sigma_1$ to $B_{N_B + 1}$ until $\abs{B_{N_B + 1}}$ is divisible by $\lceil m^{1/3}\rceil$. Construct a balanced partition $P_{N_B+1}$ of $B_{N_B + 1}$ into groups of size $\lceil m^{1/3}\rceil$ by applying
            \cref{thm:low_crossing} to the induced orders of
            $\sigma_1, \dots, \sigma_k$ on $B_{N_B+1}$. For each group $G_{N_B + 1, r^\prime}$ in $B_{N_B + 1}$, generate an independent uniformly random permutation $\pi_{N_B + 1, r^\prime}$. Mark all elements of $B_{N_B+1}$ as
            reserved, and output the first element of $\pi_{N_B+1, 1}$. Increase $N_B$ by $1$.
        \end{itemize}
    \end{itemize}
\end{itemize}
\end{mybox}

The following lemma is analogous to the validity lemmas above, so we state it without proof.

\begin{lemma}[Validity]\label{lem:rand_grouped_ps_algo_validity}
    Regardless of the realization of random permutations, the \nameref{box:rand_grouped_ps_algo} can always generate in the limit the true language $K$.
\end{lemma}

We use the setup of \cref{apx:rand_ps_algo_lower_density,apx:paired_ps_algo_lower_density}. Let $T^\star$ be a time step after which every focus is a subset of $K$; it may depend on the realization of the random permutations. Fix an oblivious adversary, let $t_r$ be the time step at which $B_r$ is created, let $\xi_r$ be the independent random seed from which the permutations of all groups of $B_r$ are drawn, and let $\mathcal{H}_r$ be the $\sigma$-algebra generated by $\xi_1, \ldots, \xi_{r-1}$. Let $F_r$ be the focus and $m_r$ the value of $m$ at time step $t_r$, and let $b_r=\lceil m_r^{2/3}\rceil$ and $g_r=\lceil m_r^{1/3}\rceil$, so that every group of $B_r$ has size $g_r$.

We identify $K$ with $\Z^+$ and define triggers, the sets $C_r$, the heads $B_r[b_r]_i$, and complete and incomplete heads within $n$ exactly as in \cref{apx:paired_ps_algo_lower_density}. In particular, \eqref{eq:paired_ps_algo_partition} holds, and whether $B_r\subseteq K$ and whether the head of $B_r$ is complete within $n$ are determined when $B_r$ is created and are therefore $\mathcal{H}_r$-measurable. Fix a prefix $[n]_i$. Each integer of $[n]_i$ belongs to $C_r$ either for a block $B_r\subseteq K$ whose head is complete within $n$ (the first case) or for any other block $B_r$ (the second case). We first count the number of blocks involved in the second case.

\begin{lemma}\label{lem:rand_grouped_ps_algo_incmp_heads} For every realization and any prefix $[n]_i$, there are at most $\log_2 n + 1$ blocks $B_r$ created after time step $T^\star$ whose heads $B_r[b_r]_i$ are incomplete within $n$ and that satisfy either $B_r\cap[n]_i\neq\emptyset$ or $B_r$ has a trigger in $[n]_i$.
\end{lemma}

\begin{proof}[Proof of \cref{lem:rand_grouped_ps_algo_incmp_heads}]
As in the proof of \cref{lem:paired_ps_algo_var_incmp_heads}, after the first block with an incomplete head, each later counted block corresponds to a unique switch loss $\ell$: either $\ell$ is an unreserved element of $[n]_i$ that triggers the new block, or the switch makes the first element of $W\cap F$ under $\sigma_i$ lie in $[n]_i$. In both cases, we charge $\ell$ to $2^{\tau^\prime}$ heads reserved by the \nameref{box:rand_grouped_ps_algo} that lie in $[n]_i$.

If there are $w$ such switch losses $\ell_1,\ell_2,\ldots,\ell_w$ with corresponding focus-change counts $\tau_1,\tau_2,\ldots,\tau_w$, then disjointness of the charged sets implies
\[
n \geq 2^{\tau_1}+2^{\tau_2}+\cdots+2^{\tau_w}.
\]
Since $\tau_1, \tau_2, \ldots, \tau_w$ are distinct positive integers, we have $n\geq 2^1+\cdots+2^w\geq 2^w$. Taking logarithms gives $w\leq \log_2 n$. The first block with an incomplete head accounts for the additional $1$.
\end{proof}

\begin{lemma}\label{lem:rand_grouped_ps_algo_few_foci}
For every realization and any prefix $[n]_i$, the blocks $B_r$ that are created after time step $T^\star$ and satisfy $B_r\cap[n]_i\neq\emptyset$ have at most $2\log_2 n+1$ distinct foci $F_r$.
\end{lemma}

\begin{proof}[Proof of \cref{lem:rand_grouped_ps_algo_few_foci}]
The proof of \cref{lem:paired_ps_algo_var_few_foci} applies verbatim. When $B_r$ is created, $(W\cap F_r)[b_r]_i\subseteq B_r\subseteq W\cap F_r$, so the head $B_r[b_r]_i$ is the set of the first $b_r$ elements of $W\cap F_r$ under $\sigma_i$; and the charging argument is the one in the proof of \cref{lem:rand_grouped_ps_algo_incmp_heads}.
\end{proof}

Next, we compute the algorithm's density among integers in the first case.

For a positive integer $m$, let
\[
b(m)=\lceil m^{2/3}\rceil
\qquad\text{and}\qquad
g(m)=\lceil m^{1/3}\rceil,
\]
and define
\[
\phi(m)=
\left(1 - \frac1{g(m)}\right) \cdot f(g(m))\cdot
\left(
1-\frac{4g(m)\cdot\left(\frac{kb(m)+g(m)}{g(m)}\right)^{1-1/k}}{b(m)}
\right).
\]

\begin{lemma}\label{lem:rand_grouped_ps_algo_cmp_heads}
For any order index $i$, any positive integer $n$, and any block $B_r$, on the event that $B_r\subseteq K$ and the head $B_r[b_r]_i$ is complete within $n$,
\[
\E{}{\abs{D \cap B_r \cap [n]_i} \mid \mathcal{H}_r} \geq \phi(m_r) \cdot \abs{C_r \cap [n]_i}.
\]
\end{lemma}

\begin{proof}[Proof of \cref{lem:rand_grouped_ps_algo_cmp_heads}]
Condition on $\mathcal{H}_r$ and suppose that $B_r\subseteq K$ and the head $B_r[b_r]_i$ is complete within $n$. Write $g = g_r$ and $\psi \defeq \left(1 - \frac1g\right) \cdot f(g)$. Let $G^\circ$ be the group whose permutation is used at time step $t_r$, let $\mathcal{G}$ be the set of groups of $B_r$ contained in $[n]_i$, and let $c$ be the number of groups of $B_r$ crossed by the cut $(B_r\cap[n]_i,B_r\setminus[n]_i)$. All of these are determined by $\mathcal{H}_r$.

We first show that, for every $G \in \mathcal{G}$,
\begin{equation}\label{eq:rand_grouped_ps_algo_per_group}
\E{}{\abs{D \cap G} \mid \mathcal{H}_r} \geq \psi \cdot \abs{G}, \qquad\text{and moreover}\qquad \E{}{\abs{D \cap G^\circ} \mid \mathcal{H}_r} \geq \psi \cdot (g + 1) \text{ if } G^\circ \in \mathcal{G}.
\end{equation}
Let $y_1,\ldots,y_g$ be the elements of $G$ in the adversary's announcement order. As in the proof of \cref{lem:rand_ps_algo_cmp_blks}, consider a bipartite graph with offline vertices $G$ and online vertices $u_1,\ldots,u_g$, where $u_j$ represents the adversary's announcement of $y_j$ and has neighbors $\{y_{j+1},\ldots,y_g\}$. If $G = G^\circ$, add an online vertex $u_0$ that represents the creation step $t_r$, arrives first, and is adjacent to all of $G$. The graph is determined by $\mathcal{H}_r$, and the algorithm's announcements within $G$ correspond exactly to the matches produced by Ranking with the permutation of $G$, which, conditional on $\mathcal{H}_r$, is uniformly distributed over the permutations of $G$. If $G \neq G^\circ$, the graph has $g$ online vertices and a matching of size $g - 1$, so \cref{thm:ranking_algo} yields $\E{}{\abs{D \cap G} \mid \mathcal{H}_r} \geq (g - 1) \cdot f(g) = \psi \cdot g$. If $G = G^\circ$, the graph has $g + 1$ online vertices and a matching of size $g$, so $\E{}{\abs{D \cap G} \mid \mathcal{H}_r} \geq g \cdot f(g + 1)$. Since $\left(1-\frac1{m+1}\right)^m = \left(1 + \frac1m\right)^{-m}$ is nonincreasing in $m$, the function $f$ is nondecreasing, and hence
\[
g \cdot f(g + 1) \geq \left(g - \frac1g\right) \cdot f(g) = \psi \cdot (g + 1) \geq \psi \cdot g.
\]
This proves \eqref{eq:rand_grouped_ps_algo_per_group}.

Next, every integer of $C_r \cap [n]_i$ is covered by \eqref{eq:rand_grouped_ps_algo_per_group}, that is, it lies in a group of $\mathcal{G}$ or is the trigger with $G^\circ \in \mathcal{G}$, unless it lies in a crossed group or is the trigger with $G^\circ \notin \mathcal{G}$. Each crossed group has at most $g - 1$ integers in $[n]_i$. Moreover, since $B_r \subseteq K$, the set $B_r \cap [n]_i$ is a prefix of $B_r$ under $\sigma_i$; since $\abs{B_r} \leq kb_r + g$ and every group has size $g$, \cref{thm:low_crossing} gives
\[
c \leq \kappa_r \defeq 4\cdot\left(\frac{kb_r+g}{g}\right)^{1-1/k}.
\]
Therefore, at most $c \cdot (g - 1) + 1 \leq \kappa_r \cdot g$ integers of $C_r \cap [n]_i$ are not covered, where we used $\kappa_r \geq 1$. Since the groups are pairwise disjoint and $\psi \geq 0$, summing \eqref{eq:rand_grouped_ps_algo_per_group} over $\mathcal{G}$ gives
\[
\E{}{\abs{D \cap B_r \cap [n]_i} \mid \mathcal{H}_r} \geq \sum_{G \in \mathcal{G}} \E{}{\abs{D \cap G} \mid \mathcal{H}_r} \geq \psi \cdot \left(\abs{C_r \cap [n]_i} - \kappa_r \cdot g\right).
\]
Finally, since the head is complete within $n$, we have $\abs{C_r \cap [n]_i} \geq b_r$, and hence
\begin{align*}
\E{}{\abs{D \cap B_r \cap [n]_i} \mid \mathcal{H}_r} &\geq \psi \cdot \left(1 - \frac{\kappa_r \cdot g}{b_r}\right) \cdot \abs{C_r \cap [n]_i}\\
&= \phi(m_r) \cdot \abs{C_r \cap [n]_i}.\qedhere
\end{align*}
\end{proof}

Since $\left(1 - \frac1{g(m)}\right)\cdot f(g(m))\to 1-1/e$ as $m\to\infty$, and since
\[
\frac{g(m)\cdot\left(\frac{kb(m)+g(m)}{g(m)}\right)^{1-1/k}}{b(m)}
\to 0
\qquad\text{as }m\to\infty,
\]
we have $\phi(m)\to 1-1/e$. Hence, for every $\varepsilon > 0$, there exists an integer $T=T(\varepsilon)$ such that $\phi(m)\geq 1-1/e-\varepsilon$ for every $m\geq T$. It remains to control the contribution of blocks with $m_r<T$.

\begin{lemma}\label{lem:rand_grouped_ps_algo_small_heads}
Fix any $\varepsilon > 0$, and let
\[
C_\varepsilon\defeq k(T-1)\cdot\left(k\lceil T^{2/3}\rceil+\lceil T^{1/3}\rceil+1\right).
\]
For every realization and any prefix $[n]_i$, the blocks $B_r$ that are created after time step $T^\star$ and satisfy $B_r \cap [n]_i \neq \emptyset$ and $m_r < T$ satisfy
\[
\sum_{r} \abs{C_r \cap [n]_i} \leq C_\varepsilon\cdot(2\log_2 n+1),
\]
where the sum ranges over these blocks.
\end{lemma}

\begin{proof}[Proof of \cref{lem:rand_grouped_ps_algo_small_heads}]
At the time $B_r$ is created, some order $j\in[k]$ satisfies $m_r=\rho_j^{F_r}(W\cap F_r)$, and the corresponding first available element of $F_r$ under $\sigma_j$, whose rank within $F_r$ under $\sigma_j$ equals $m_r$, is reserved in $B_r$. Hence, for each language $F$, each order $j\in[k]$, and each value of $m_r$, this can happen for at most one block with focus $F$. Thus, for a fixed language $F$, there are at most $k(T-1)$ blocks with focus $F$ and $m_r<T$. Each such block satisfies
\[
\abs{C_r} \leq k\lceil T^{2/3}\rceil+\lceil T^{1/3}\rceil+1,
\]
including its trigger. Hence the blocks with focus $F$ and $m_r<T$ contribute at most $C_\varepsilon$ to the sum. By \cref{lem:rand_grouped_ps_algo_few_foci}, the blocks in the sum have at most $2\log_2 n+1$ distinct foci. The result follows.
\end{proof}

Finally, we can compute the lower density of the \nameref{box:rand_grouped_ps_algo}.

\begin{proof}[Proof of \cref{thm:rand_grouped_ps_algo_lower_density}]
Fix any $\varepsilon \in (0, 1 - 1/e)$, let $T = T(\varepsilon)$ and $C_\varepsilon$ be as above, and fix an order index $i$ and an arbitrary prefix $[n]_i$. For each block $B_r$, let $I_r \defeq \mathbf{1}\set{B_r \subseteq K \text{ and } B_r[b_r]_i \subseteq [n]_i}$ indicate that $B_r \subseteq K$ and the head $B_r[b_r]_i$ is complete within $n$, and let $J_r \defeq \mathbf{1}\set{m_r \geq T}$; both are $\mathcal{H}_r$-measurable. Let
\[
\Delta_{i,n} \defeq \sum_r (1 - I_r) \cdot \abs{C_r \cap [n]_i}
\qquad\text{and}\qquad
S_{i,n} \defeq \sum_r I_r \cdot (1 - J_r) \cdot \abs{C_r \cap [n]_i}
\]
be the number of integers in the second case and the number of integers in the first case associated with blocks with $m_r < T$, respectively. By \eqref{eq:paired_ps_algo_partition}, $\sum_r I_r \cdot J_r \cdot \abs{C_r \cap [n]_i} = n - \Delta_{i,n} - S_{i,n}$.

\begin{enumerate}
    \item For the integers in the first case: since the blocks are pairwise disjoint, we have $\mu_{i,n}(D \cap K) \geq \sum_r I_r \cdot J_r \cdot \abs{D \cap B_r \cap [n]_i}$. Since $I_r$, $J_r$, $m_r$, and $\abs{C_r \cap [n]_i}$ are $\mathcal{H}_r$-measurable, the law of total expectation and \cref{lem:rand_grouped_ps_algo_cmp_heads} give, for every $r$,
    \begin{align*}
    \E{}{I_r \cdot J_r \cdot \abs{D \cap B_r \cap [n]_i}} &= \E{}{I_r \cdot J_r \cdot \E{}{\abs{D \cap B_r \cap [n]_i} \mid \mathcal{H}_r}}\\
    &\geq \left(1 - \frac1e - \varepsilon\right) \cdot \E{}{I_r \cdot J_r \cdot \abs{C_r \cap [n]_i}},
    \end{align*}
    where we used $\phi(m_r) \geq 1 - 1/e - \varepsilon$ whenever $J_r = 1$. Summing over $r$, and using monotone convergence for these series of nonnegative terms, we obtain
    \begin{equation}\label{eq:rand_grouped_ps_algo_first_case}
    \E{}{\mu_{i,n}(D \cap K)} \geq \left(1 - \frac1e - \varepsilon\right) \cdot \E{}{n - \Delta_{i,n} - S_{i,n}}.
    \end{equation}
    \item For the remaining integers, fix a realization, and let $M_0$ be the total number of integers in the sets $C_r$ of the blocks created up to time step $T^\star$. By \cref{lem:rand_grouped_ps_algo_validity}, $M_0$ is finite and depends on neither $i$ nor $n$.

    Every block created after time step $T^\star$ is contained in $K$, and hence contributes to $\Delta_{i,n}$ only if its head is incomplete within $n$ and $C_r \cap [n]_i \neq \emptyset$. By \cref{lem:rand_grouped_ps_algo_incmp_heads}, there are at most $\log_2 n+1$ such blocks. For each such block, the integers it contributes consist of at most one trigger, plus the elements of $B_r\cap[n]_i$. If $B_r\cap[n]_i\neq\emptyset$, then, as in the proof of \cref{thm:paired_ps_algo_lower_density}, $m_r\leq\rho_i^{F_r}(W\cap F_r)\leq n$, so $b_r\leq\lceil n^{2/3}\rceil$, $g_r\leq\lceil n^{1/3}\rceil$, and
    \[
    \abs{B_r}\leq k\lceil n^{2/3}\rceil+\lceil n^{1/3}\rceil.
    \]
    Otherwise, the block contributes only its trigger. Either way, such a block contributes at most $k\lceil n^{2/3}\rceil+\lceil n^{1/3}\rceil+1$ integers.

    Every block created after time step $T^\star$ contributes to $S_{i,n}$ only if its head is complete within $n$, and hence it intersects $[n]_i$, and $m_r < T$. By \cref{lem:rand_grouped_ps_algo_small_heads}, these blocks contribute at most $C_\varepsilon \cdot (2\log_2 n + 1)$ integers. Therefore,
    \[
    \Delta_{i,n} + S_{i,n} \leq M_0 + \left(k\lceil n^{2/3}\rceil+\lceil n^{1/3}\rceil+1\right)\cdot(\log_2 n+1) + C_\varepsilon\cdot(2\log_2 n+1),
    \]
    and hence $(n - \Delta_{i,n} - S_{i,n})/n \to 1$ as $n \to \infty$, for every realization.
\end{enumerate}

Note that $T^\star$, and hence $M_0$, may depend on the algorithm's random choices, so we pass to the limit using Fatou's lemma. Since $\mu_{i,n}(K) = n$ and $(n - \Delta_{i,n} - S_{i,n})/n \geq 0$, \eqref{eq:rand_grouped_ps_algo_first_case} and Fatou's lemma yield that the algorithm's lower density under the $k$ orders is at least
\begin{align*}
\min_{i\in[k]}\liminf_{n\to\infty}\frac{\E{}{\mu_{i,n}(D\cap K)}}{\mu_{i,n}(K)}
&\geq
\min_{i\in[k]}\left(1-\frac1e-\varepsilon\right) \cdot \liminf_{n \to \infty}\E{}{\frac{n - \Delta_{i,n} - S_{i,n}}{n}}\\
&\geq
\min_{i\in[k]}\left(1-\frac1e-\varepsilon\right) \cdot \E{}{\liminf_{n \to \infty}\frac{n - \Delta_{i,n} - S_{i,n}}{n}}\\
&=1-\frac1e-\varepsilon.
\end{align*}
Since $\varepsilon \in (0, 1 - 1/e)$ was arbitrary, the lower density is at least $1-1/e$.
\end{proof}

\end{document}